\documentclass[10pt]{amsart}

\usepackage{macros,amsmath}

\title{Renormalization for holomorphic field theories}

\author{Brian R. Williams}
\address{Northeastern University}
\email{brwilliams@northeastern.edu}
\urladdr{}

\date{}

\begin{document}

\maketitle

\begin{abstract}
We introduce the concept of a holomorphic field theory on any complex manifold in the language of the Batalin-Vilkovisky formalism. 
When the complex dimension is one, this setting agrees with that of chiral conformal field theory. 
Our main result concerns the behavior of holomorphic theories under renormalization group flow.
Namely, we show that holomorphic theories are one-loop finite. 
We use this to completely characterize holomorphic anomalies in any dimension.
Throughout, we compare our approach to holomorphic field theories to more familiar approaches including that of supersymmetric field theories. 
\end{abstract}

\tableofcontents 

\section{Introduction}

From a mathematical perspective, much of the appeal of quantum field theory is that often theories depend naturally on input geometric data. 
Furthermore, the usual quantities in quantum field theory such as expectation values and the partition function produce invariants of these underlying geometries.
An important and fruitful instance of this is the notion of a {\em topological field theory}.
Mathematically, a topological field theory can be defined on an arbitrary manifold of a fixed dimension.
In a precise way, topological theories depend naturally on the smooth structure of the manifold (or smooth structures on associated data such as a bundle).
A more complicated class of theories are Riemannian field theories, which, in addition to smooth structures, are sensitive to input metric data. 
These theories have more refined invariants associated to them, such as the $\beta$-function, and are often more relevant to physical examples.
In this paper, we study a class of theories that lie between the aforementioned examples. 
These {\em holomorphic theories} depend naturally on the complex structure of the underlying space-time. 

The idea of studying holomorphic dependence in quantum field theory is certainly not a new one.
The most well-known case of this appears in complex dimension one with the notion of a chiral conformal field theory.
Here, the holomorphic structure shines most brightly through the {\em operator product expansion} (OPE) of chiral operators in the theory. 
This says that the dependence on the product of operators on their relative position is holomorphic, even at the quantum level. 
These operators combine to form a mathematical object called a vertex algebra.
Numerous calculations in conformal field theory reduce to algebraic manipulations at the level of vertex algebras. 
Furthermore, on arbitrary algebraic curves, the phenomena of operator product expansions has been interpreted mathematically through the pioneering work of Beilinson and Drinfeld on chiral algebras \cite{BD}. 
This is arguably one of the greatest successes of mathematics in describing a small, albeit important, class of field theories. 

Past dimension one, in complex dimensions two, four and six, an approach to studying special types of holomorphic theories has appeared in the work of Nekrasov and collaborators \cite{NekThesis, NekChiral, NekCFT}. 
In the physics literature, a holomorphic version of Chern-Simons theory has appeared as a twist of 10d supersymmetric Yang-Mills theory \cite{Baulieu} as well as in \cite{Popov1, Popov2}.

A holomorphic theory of gravity has been proposed in \cite{BCOVafa}.
Recently, Costello and Li \cite{bcov, CL1,CosKos1, CosM} have given a mathematical construction of this theory in the BV formalism and have studied applications to the topological string.
Our motivation for the definition of a holomorphic quantum field theory is largely based off of an abstraction of the formalism that Costello and Li have developed in their work. 
Many of the technical methods that are employed here are generalizations of some key results that appear in this work. 

The goal of this paper is two-fold. 
In the first part the discussion is fairly formal.
After a short recollection of field theory in the Batalin-Vilkovisky formalism, we go on to define the definition of a holomorphic field theory on any complex manifold. 
We characterize holomorphic deformations of holomorphic theories and provide numerous examples of these theories in the language we set up.

The second part of the paper proceeds to study quantizations of holomorphic field theories defined on $\CC^d$, for any $d \geq 1$. 
Of course, studying properties of quantization is extremely theory-dependent.
Nevertheless, our main result says that when it comes to renormalization, holomorphic theories are generically well-behaved.
We show that the renormalization of a holomorphic theory on $\CC^d$ is {\em finite} for quantization at one-loop.
A more precise statement is given in Theorem \ref{thm: holrenorm3}. 
 
The approach to quantum field theory we use follows Costello's theory of renormalization and the Batalin-Vilkovisky formalism developed in \cite{CosRenorm}.
In broad strokes, it says that to construct a full quantum field theory it suffices to define the theory at each energy (or length) scale and to ask that these descriptions be compatible as we vary the scale.
Concretely, this compatibility is through the {\em renormalization group (RG) flow} and is encoded by an operator $W(P_{\epsilon < L}, -)$ acting on the space of functionals. 
The functional $W(P_{\epsilon < L},-)$ is defined as a sum over weights of graphs which is how Feynman diagrams appear in Costello's formalism.
The infamous infinities of quantum field theory arise due to studying behavior of theories at arbitrarily high energies (or small lengths). 
In physics this is called the ultra-violet (UV) divergence.
Our result can be interpreted by saying that, at one-loop, holomorphic theories have no UV divergences. 

Although we do not consider this topic in the present paper, a large collection of examples of holomorphic theories come from familiar physical theories.
Namely, holomorphic theories generically appear as minimal {\em twists} of supersymmetric theories. 
These are more general than the topological twists considered by Witten in \cite{WittenTwist}.
Any supercharge $Q$ of a supersymmetric theory satisfying $Q^2 = 0$ allows one to construct a ``twist". 
In some cases, where Clifford multiplication with $Q$ spans all translations such a twist becomes a topological theory (in the weak sense). 
In any case, however, such a $Q$ defines a ``holomorphic twist" \cite{CostelloHolomorphic}, which results in the type of holomorphic theories we consider.
Regularization in supersymmetric theories, especially gauge theories, is notoriously difficult. 
Our result implies that after twisting the analytic difficulties become much easier to deal with. 
Consequently, facets of these theories, such as their anomalies, can be cast in a more algebraic framework.
For a recent discussion of holomorphic aspects of twists of supersymmetric theories see \cite{EagSab}. 

In no way does this paper tell the complete story of holomorphic field theory.
A major future program of the author is to study the behavior of operators for holomorphic field theory, even in the case that the complex manifold is $X = \CC^d$. 
In general, the operators of any quantum field theory form a {\em factorization algebra} \cite{CG1,CG2}. 
For one-dimensional holomorphic theories, our formalism recovers the theory of chiral and vertex algebras \cite{BWVir, GGW, CG1}. 
When $d \geq 2$ there is strong evidence that the factorization algebras of holomorphic theories combine to form some higher dimensional vertex algebra structure, where the OPE still varies holomorphically with respect to the relative location of the operators. 
We will return to this in later publications.  

\subsection*{Notation and homological conventions}

\begin{itemize}
\item Vector spaces and cochain complexes are defined over $\CC$. 
All tensor products $\tensor = \tensor_{\CC}$ are defined over $\CC$, unless otherwise specified. 
\item If $V^*$ is a graded vector space, then $V^*[k]$ is the graded vector space which is $V^{i + k}$ in degree $i$. 
\item If $V^*$ is a graded vector space, the notation $v \in V^* [k]$ will signify that $v$ is a homogenous element of degree $k$.

\item All products and commutators $[-,-]$ will be understood as products and commutators in the graded sense, and will hence follow the Koszul rule of signs. 
For instance, if $v$ is degree $i$ and $w$ is degree $j$ then $[v,w] = -(-1)^{ij} [w,v]$ whenever $[v,w]$ and $[w,v]$ are defined.

\end{itemize} 
\subsection{Acknowledgements}

First, and foremost, the author would like to express his gratitude to Si Li for his guidance and advice in all things related to holomorphic methods in quantum field theory. 
His work in \cite{LiFeynman} and the work of Si Li and Kevin Costello on BCOV theory in \cite{bcov} were the main motivations for this work. 
The author learned the methods of renormalization employed here from the work of Costello and Li in \cite{bcov} and by Li in \cite{LiFeynman}.
%\cite{\brian{finish guidance, advice, inspiration, motivation,}.
The author would also like to thank Owen Gwilliam for comments and suggestions he made on a previous version of this paper that appeared in the author's thesis. 
Also, the author thanks Matt Szczesny for discussions related to, and comments made on, a more recent version of this paper. 
The author would also like to thank Northwestern University, where he received support as a graduate student whilst most of this work took place.
In addition, the author enjoyed support as a graduate student research fellow under Award DGE-1324585. 
 
\section{The definition of a holomorphic field theory}

The goal of this section is to define the notion of a holomorphic field theory. 
This is a variant of Costello's definition of a theory in the Batalin-Vilkovisky formalism, which we will recall at a rapid pace in the first part of this section.
In crude summary, to arrive at the definition of a holomorphic field theory we modify the definition of an ordinary BV theory by inserting the word ``holomorphic" in front of most objects (bundles, differential operators, etc..).
By applying the Dolbeault complex in appropriate locations, we will recover Costello's definition of a theory, but with a holomorphic flavor, see Table \ref{table: holtoBV}. 

\subsection{A recollection of the BV-BRST formalism}

In this section we will give an expedient review of the classical Batalin-Vilkovisky formalism.
We will also set up the requisite conventions and notations that we will use throughout this paper. 

\subsubsection{Classical field theory} \label{sec: classical bv}

Classical field theory is a formalism for describing a physical system in terms of objects called {\em fields}. 
Mathematically, the space of fields is a (most often infinite dimensional) vector space $\sE$. 
Classical physics is described by the critical locus of a (usually real or complex valued) linear functional on the space of fields 
\be\label{actionfnl}
S : \sE \to \RR \;\; {\rm or} \;\; \CC,
\ee
called the {\em action functional}. 
The critical locus is the locus of fields that have zero variation
\be
{\rm Crit}(S) := \{\varphi \in \sE \; | \; \d S (\varphi) = 0\} .
\ee
A field $\varphi$ satisfying the equation $\d S (\varphi) = 0$ is said to be a {\em solution to the classical equations of motion}. 

Even in the finite dimensional case, if the functional $S$ is not sufficiently well-behaved the critical locus can be still be highly singular. 
The starting point of the classical Batalin-Vilkovisky formalism is to instead consider the {\em derived} critical locus.
To get a feel for this, we review the finite dimensional situation.
Let $M$ be a manifold, which is our ansatz for $\sE$ at the moment, and suppose $S : M \to \RR$ is a smooth map.
The critical locus is the intersection of the graph of $\d S$ in $T^*M$ with the zero section $0 : M \to T^*M$.
Thus, functions on the critical locus are of the form
\ben
\sO({\rm Crit}(S)) = \sO(\Gamma(\d S)) \tensor_{\sO(T^*M)} \sO(M) .
\een
The derived critical locus is a derived space whose dg ring of functions is 
\ben
\sO({\rm Crit}^{h}(S)) = \sO(\Gamma(\d S)) \tensor^{\LL}_{\sO(T^*M)} \sO(M).
\een
We have replaced the strict tensor product with the derived one.
Using the Koszul resolution of $\sO(M)$ as a $\sO(T^*M)$-module one can write this derived tensor product as a complex of polyvector fields equipped with some differential:
\ben
\sO({\rm Crit}^h(S)) \simeq \left({\rm PV}^{-*}(M), \iota_{\d S}\right) .
\een
In cohomological degree $-i$ we have ${\rm PV}^{-i} (M) = \Gamma(M, \wedge^i TM)$ and $\iota_{\d S}$ denotes contraction with the one-form $\d S$ (which raises cohomological degree with our regrading convention).
With our grading convention we have $\sO(T^*[-1] M) = {\rm PV}^{-*}(M)$. 
The space $\sO(T^*[-1]M)$ has natural shifted Poisson structure, which takes the form of the familiar Schouten-Nijenhuis bracket of polyvector fields.

The takeaway is that the derived critical locus of a functional $S : M \to \RR$ has the structure of a $(-1)$-shifted symplectic space.
This will be the starting point for our definition of a theory in the BV formalism in the general setting.
 
In all non-trivial examples the space of fields $\sE$ is infinite dimensional and we must be careful with what functionals $S$ we allow.
The space of fields we consider will always have a natural topology, and we will choose functionals that are continuous with respect to it. 
We include a discussion of our convention for infinite dimensional vector spaces including duals and spaces of functionals in the Appendix. 

In general, the space of fields of a field theory is equal to the space of smooth sections of a $\ZZ$-graded vector bundle $E \to X$ on a manifold $\sE = \Gamma(X, E)$. 
The $\ZZ$-grading is the cohomological, or BRST \footnote{Named after Becchi, Rouet, Stora, Tyutin, for which our approach to field theory is greatly influenced by their original mathematical approach to quantization.}, grading of the theory.

\subsubsection{Local functionals}

The class of functionals $S : \sE \to \RR$ defining the classical theories we consider are required to be {\em local}, or given by the integral of a Lagrangian density. 
We define this concept now.

Let $D_X$ denote the sheaf of smooth differential operators on $X$. 
If $E$ is any graded vector bundle on $X$ let ${\rm Jet}(E)$ denote its bundle of $\infty$-jets. 
This is a smooth vector bundle, albeit infinite rank, on $X$ whose fiber over $y \in X$ can be identified with
\ben
E_y \times \CC[[x_1, \ldots, x_n]] .
\een
Here, $\{x_i\}$ is a formal coordinate near $y$. 
This object is given the natural structure of a pro object in the category of vector bundles.
We let $J(E)$ denote the associated sheaf of smooth sections.
It is well-known that ${\rm Jet}(E)$ is equipped with a natural flat connection rendering $J(E)$ with the structure of a smooth $D_X$-module.

In the Appendix we define the algebra of functions $\sO(\sE(X))$ on the space of global sections $\sE(X)$.
This is the completed symmetric algebra on the linear dual of $\sE(X)$, where the tensor product and dual are interpreted in the appropriate topological sense. 
Likewise, there is the space of reduced functionals $\sO_{red}(\sE(X)) = \sO(\sE(X)) / \RR$. 
It is the quotient of all functionals by the constant polynomial functions. 

The space $\sO_{red}(J(E))$ inherits a natural $D_X$-module structure from $J(E)$. 
We refer to $\sO_{red}(J(E))$ as the space of {\em Lagrangians} on the vector bundle $E$. 
Every element $F \in \sO_{red}(J(E))$ can be expanded as $F = \sum_n F_n$ where each $F_n$ is an element 
\ben
F_n \in {\rm Hom}_{C^\infty_X} (J(E)^{\tensor n}, C^\infty_X)_{S_n} \cong {\rm PolyDiff}(\sE^{\tensor n}, C^\infty(X))_{S_n}
\een
where the right-hand side is the space of polydifferential operators.
The proof of the isomorphism on the right-hand side can be found in Chapter 5 of \cite{CosRenorm}.
We refer to $\sO_{red}(J(E))$ as the (left) $D_X$-module of {\em Lagrangians} on the vector bundle $E$. 

A local functional is given by a Lagrangian densities modulo total derivatives.
The mathematical definition is the following.

\begin{dfn} \label{dfn: local fnl}
Let $E$ be a graded vector bundle on $X$.
Define the sheaf of {\em local functionals} on $X$ to be
\ben
\oloc(\sE) = {\rm Dens}_X \tensor_{D_X} \sO_{red}(J(E)),
\een
where we use the natural right $D_X$-module structure on densities.
\end{dfn}

Note that we always consider local functionals coming from Lagrangians modulo constants. 
We will not be concerned with local functions associated to constant Lagrangians. 

From the expression for functionals in Lemma \ref{lem: fnls} we see that integration defines an inclusion of sheaves
\be\label{local inclusion}
i : \oloc(\sE) \hookrightarrow \sO_{red}(\sE_c) .
\ee
Often times when we describe a local functional we will write down its value on test compactly supported sections, then check that it is given by integrating a Lagrangian density, which amounts to lifting the functional along $i$. 

\subsubsection{The definition of a classical field theory}

Before giving the definition, we need to recall what the proper notion of a shifted symplectic structure is in the geometric setting that we work in.

\begin{dfn}\label{dfn: symplectic}
Let $E$ be a graded vector bundle on $X$.
A $k$-{\em shifted symplectic structure} is an isomorphism of graded vector spaces
\ben
E \cong_{\omega} E^![k] = \left({\rm Dens}_X \tensor E^\vee\right)[k]
\een
that is graded anti-symmetric.
\end{dfn}

If $\omega^*$ is the formal adjoint of the isomorphism $\omega^* : E \cong E^![k]$, anti-symmetry amounts to the condition $\omega^* = - \omega$. 
In general, $\omega$ does {\em not} induces a Poisson structure on the space of all functionals $\sO(\sE)$. 
This is because, as we have seen above, elements of this space are given by distributional sections and hence we cannot pair elements with overlapping support.
The symplectic structure does, however, induce a Poisson bracket on {\em local} functionals. \footnote{Note that $\oloc(\sE)$ is not a shifted Poisson algebra since there is no natural commutative product.}
We will denote the bracket induced by a shifted symplectic structure by $\{-,-\}$. 

We are now ready to give the precise definition of a classical field theory.

\begin{dfn}[\cite{CG2} Definition 5.4.0.3] \label{dfn: classical}
A {\em classical field theory} in the BV formalism on a smooth manifold $X$ is a $\ZZ$-graded vector bundle $E$ equipped with a $(-1)$-shifted symplectic structure together with a local functional $S \in \oloc(\sE)$ such that:
\begin{enumerate}
\item the functional $S$ satisfies the {\em classical master equation} 
\ben
\{S, S\} = 0;
\een
\item $S$ is at least quadratic, so we can write it (in a unique way) as 
\ben
S(\varphi) = \omega(\varphi, Q \varphi) + I(\varphi)
\een
where $Q$ is a linear differential operator such that $Q^2 = 0$, and  $I \in \oloc(\sE)$ is at least cubic;
\item the complex $(\sE, Q)$ is elliptic.
\end{enumerate}
\end{dfn}

In the physics literature, the operator $Q$ is known as the linearized BRST operator, and $\{S,-\} = Q + \{I,-\}$ is the full BRST operator.
Ellipticity of the complex $(\sE,Q)$ is a technical requirement that will be very important in our approach to the issue of renormalization in perturbative quantum field theory.
The classical master equation is equivalent to
\ben
Q I + \frac{1}{2} \{I,I\} = 0 .
\een

A {\em free theory} is a classical theory with $I = 0$ in the notation above. 
Thus, a free theory is a simply an elliptic complex equipped with a $(-1)$-shifted symplectic pairing where the differential in the elliptic complex is graded skew-self adjoint for the pairing.  

Although the space $\sO(\sE)$ does not have a well-defined shifted Possoin bracket induced from the symplectic pairing, the operator $\{S,-\} : \sO(\sE) \to \sO(\sE)[1]$ {\em is} well-defined since $S$ is local by assumption. 
By assumption, it is also square zero. 
The complex of global classical observables of the theory is defined by
\ben
\Obs^{\cl}_{\sE}(X) = (\sO(\sE(X)), \{S,-\}) .
\een
This complex is the field theoretic replacement for functions on the derived locus of $S$ from the beginning of this section.
Although it does not have a $P_0$-structure, there is a subspace that does. 
This is sometimes referred to as the {\em BRST} complex in the physics literature.

%Suppose that $E$ is the graded vector bundle underlying the data of a free theory in the BV formalism.
%It is shown in \brian{ref} \cite{CG2} that the inverse of the $(-1)$-shifted symplectic form induces a bracket $\{-,-\}$ of cohomological degree $+1$ on local functionals $\oloc(\sE)$. 
%This bracket satisfies a graded version of the Jacobi identity, and is compatible with the BV differential $Q$.
%In summary, the shift of local functionals $\oloc(\sE)[-1]$ inherits the structure of a dg Lie algebra. 

\subsubsection{A description using $L_\infty$ algebras}

There is a completely equivalent way to describe a classical field theory that helps to illuminate the mathematical meaningfulness of the definition given above. 
The requisite concept we need to introduce is that of a {\em local Lie algebra} (or local $L_\infty$ algebra).

First, recall that an $L_\infty$ algebra is a modest generalization of a dg Lie algebra where the Jacobi identity is only required to hold up to homotopy.
The data of an $L_\infty$ algebra is a graded vector space $V$ with, for each $k \geq 1$, a $k$-ary bracket
\ben
\ell_k : V^{\tensor k} \to V[2-k]
\een
of cohomological degree $2-k$. 
These maps are required to satisfy a series of conditions, the first of which says $\ell_1^2 = 0$.
The next says that $\ell_2$ is a bracket satisfying the Jacobi identity up to a homotopy given by $\ell_3$.
For a detailed definition see we refer the reader to \cite{StasheffDG, GetzlerLie}.

We now give the definition of a local $L_\infty$ algebra on a manifold $X$.
This has appeared in Chapter 4 of \cite{CG2}. 

\begin{dfn}
A {\em local $L_\infty$ algebra} on $X$ is the following data:
\begin{itemize}
\item[(i)] a $\ZZ$-graded vector bundle $L$ on $X$, whose sheaf of smooth sections we denote $\sL^{sh}$, and
\item[(ii)] for each positive integer $n$, a polydifferential operator in $n$ inputs
\ben
\ell_n : \underbrace{\sL^{sh} \times \cdots \times \sL^{sh}}_{\text{$n$ times}} \to \sL^{sh}[2-n]
\een
\end{itemize}
such that the collection $\{\ell_n\}_{n \in \NN}$ satisfy the conditions of an $L_\infty$ algebra.
In particular, $\sL$ is a sheaf of $L_\infty$ algebras. 
\end{dfn}

The simplest example of a local Lie algebra starts with the data of an ordinary Lie algebra $\fg$. 
We can then take the constant bundle $\ul{\fg}_X$ with fiber $\fg$. 
The Lie bracket on $\fg$ extends to define the structure of a local Lie algebra.
In this case, the sheaf of Lie algebras is $C^\infty_X \tensor \fg$.  
Another important example of a local Lie algebra is given by the Lie algebra of vector fields ${\rm Vect}(X)$ on a smooth manifold. 
The Lie bracket of vector fields is a bidifferential operator on the tangent bundle and this equips the sheaf of sections with the structure of a sheaf of Lie algebras.
%We will study the holomorphic version of this local Lie algebra in Chapter \ref{chap: symmetries}.

Just as in the case of an ordinary graded vector bundle, we can discuss local functionals on a local Lie algebra $L$. 
In this case, the $L_\infty$ structure maps give this the structure of a sheaf of complexes, providing a local version of the Chevalley-Eilenberg cochain complex. 
Indeed, the $\infty$-jet bundle $J L$ is an $L_\infty$ algebra object in $D_X$-modules and so we can define the $D_X$-module of reduced Chevalley-Eilenberg cochains $\cred^*(J L)$. 
Mimicking the definition above, we arrive at the following local version of Lie algebra cohomology that will come up again and again in this thesis.

\begin{dfn}
Let $L$ be a local Lie algebra. 
The local Chevalley-Eilenberg cochain complex is the sheaf of cochain complexes
\ben
\cloc^*(\sL) = {\rm Dens}_X \tensor_{D_X} \cred^*(L) .
\een
We denote the global sections by $\cloc^*(\sL(X))$. 
\end{dfn}

\begin{rmk}
Concretely, a section $I$ of $\cloc^*(\sL)$ supported on $U \subset X$ is a sum of cochains of the form
\[
\phi \mapsto \int_U (D_{1} \phi_1 \cdots  D_{n} \phi_n)  \; \dvol_X
\]
where $\phi$ is a compactly supported section of $L$ over $U$ and where $D_i$ are differential operators $\sL \to C^\infty_X$
\end{rmk}

The {\em local cohomology} of a local Lie algebra is the cohomology of the local CE complex, which we will denote $H^*_{\rm loc}(\sL(X))$. 

\begin{rmk}
We have already remarked that for a graded vector bundle $E$ there is an embedding $\oloc(\sE) \hookrightarrow \sO_{red}(\sE)$.
This translates to an embedding of sheaves of cochain complexes $\cloc^*(\sL) \hookrightarrow \cred^*(\sL_c)$ for any local Lie algebra $\sL$. 
In the case of vector fields, there is a related cochain complex that has been studied extensively in the context of characteristic classes of foliations \cite{Fuks, Guillemin, LosikDiag, Bernstein}. 
Suppose, for simplicity, that $X$ is a compact smooth manifold.
If ${\rm Vect}(X)$ is the Lie algebra of vector fields on $X$ then the (reduced) {\em diagonal cochain complex} is the subcomplex 
\ben
{\rm C}^*_{\Delta,\rm red}({\rm Vect}(X)) \subset \cred^*({\rm Vect}(X))
\een
consisting of cochains $\varphi : {\rm Vect}(X)^{\tensor k} \to \CC$ satisfying $\varphi(X_1,\ldots,X_k) = 0$ if $\bigcap_{i=1}^k {\rm Supp}(X_i) = \emptyset$. 
That is, the cocycle vanishes unless all of the supports of the inputs overlap nontrivially. 
The inclusion of the local cochain complex $\cloc^*({\rm Vect}(X)) \subset \cred^*({\rm Vect}(X))$ factors through this subcomplex to give a sequence of inclusions
\ben
\cloc^*({\rm Vect}(X)) \hookrightarrow {\rm C}^*_{\Delta,\rm red}({\rm Vect}(X)) \hookrightarrow \cred^*({\rm Vect}(X)) .
\een
This is because the cochain of ${\rm Vect}(X)$ defined from a local cochain involves the integral of local operators applied to the inputs.
\end{rmk}

It turns out that the definition of a classical field theory can be repackaged in terms of certain structures on a local $L_\infty$ algebra.
The first piece of data we need to transport to the $L_\infty$ side is that of a symplectic pairing. 
The underlying data of a local $L_\infty$ algebra $L$ is a graded vector bundle. 
In Definition \ref{dfn: symplectic} we have already defined a $k$-shifted symplectic pairing. 
On the local Lie algebra sign, we ask for $k=-3$ shifted symplectic structures that are also invariant for the $L_\infty$ structure maps. 

Also, an important part of a classical field theory is ellipticity. 
We say a local $L_\infty$ algebra is {\em elliptic} if the complex $(\sL, \d = \ell_1)$ is an elliptic complex.

\begin{prop}[\cite{CG2} Proposition 5.4.0.2]
The following structures are equivalent:
\begin{enumerate}
\item a classical field theory in the BV formalism $(\sE, \omega, S)$;
\item an elliptic local Lie algebra structure on $L = E [1]$ equipped with a $(-3)$-shifted symplectic pairing.
\end{enumerate}
\end{prop}

\begin{proof} (Sketch) 
The underlying graded vector bundle of the space of fields $\sE$ is $E$ and we obtain the bundle underlying the local $L_\infty$ algebra by shifting this down $L = E[1]$. 
The $(-1)$-shifted symplectic structure on $E$ transports to a $(-3)$-shifted on on $L$. 
The $L_\infty$ structure maps for $L$ come from the Taylor components of the action functional $S$. 
The exterior derivative of $S$ is a section
\ben
\d S \in \cloc^*(\sL, \sL^![-1]),
\een
where on the right-hand side we have zero differential.
The Taylor components are of the form $(\d S)_n : \sL^{\tensor n} \to \sL^![-1]$. 
Using the shifted symplectic pairing we can identify these Taylor components with maps $(\d S)_n : \sL^{\tensor n} \to \sL[2]$. 
Thus, $\d S$ can be viewed as a section of $\cloc^*(\sL, \sL[2])$. 
This is precisely the space controlling deformations of $\sL$ as a local Lie algebra.
One checks immediately that the classical master equation is equivalent to the fact that $\d S$ is a derivation, hence it determines the structure of a local Lie algebra. 
The first Taylor component $\ell_1$ is precisely the operator $Q$ before, so ellipticity of $(\sE, Q)$ is equivalent to ellipticity of $(\sL, \ell_1)$. 
\end{proof}

\subsection{Free holomorphic field theories}

In this section we proceed with the general definition of a classical holomorphic field theory. 
In complex dimension one, this definition has appeared in Section 3.2 of \cite{LiVertex}, where it was referred to as a ``two-dimensional chiral theory". 
The formulation here can be seen as a straightforward generalization of the definition of a chiral theory on a Riemann surface to arbitrary complex manifolds. 

Throughout this section, we fix a complex manifold $X$ of complex dimension $d$. 
We start with the definition of a {\em free} holomorphic field theory on $X$, from there we will go on to describe how to incorporate interactions. 

The essential information that governs a classical field theory are its equations of motion. 
For a free theory, the equations of motion are linear in the space of fields.
At least classically, the setting of free theories can essentially be reduced to the study linear partial differential equations.

First, we must come to terms with the fields of a holomorphic theory. 
Just as in the case of an ordinary field theory, they will arise as sections of some $\ZZ$-graded vector bundle on $X$.
The $\ZZ$-grading plays the same role as in the usual setting, it counts the BRST, or ghost, degree. 
We will also refer to this as the cohomological degree.
For a {\em holomorphic theory} the crucial step is that we impose that this graded vector bundle be holomorphic.  
By a holomorphic $\ZZ$-graded vector bundle we mean a $\ZZ$-graded vector bundle $
V^\bullet = \oplus_i V^i [-i]$ (which we will usually abbreviate simply as $V$) such that each graded piece $V^i$ is a holomorphic vector bundle (here $V^i$ is in cohomological degree $+i$).
Thus, in order to define a holomorphic field theory on a complex manifold $X$ we start with the data:

\begin{itemize}
\item[(1)] a $\ZZ$-graded holomorphic vector bundle $V^\bullet = \oplus_i V^i [-i]$ on $X$, so that the finite dimensional holomorphic vector bundle $V^i$ is in cohomological degree $i$. 
\end{itemize}

\begin{rmk}
For supersymmetric theories it may be desirable to include an additional $\ZZ/2$, or fermionic, grading into the data of the space of fields, but we do not consider that here.
\end{rmk}

A free classical theory is made up of a space of fields as above together with the data of a linearized BRST differential $Q^{BRST}$ and a shifted symplectic pairing of cohomological degree $-1$. 
Ordinarily, the BRST operator is simply a differential operator on the underlying vector bundle defining the fields. 
For the class of theories we are considering, we require this operator be holomorphic. 
For completeness, we briefly recall this notion.

Suppose that $E$ and $F$ are two holomorphic vector bundles on $X$.
Note that the Hom-bundle ${\rm Hom}(E,F)$ inherits a natural holomorphic structure. 
By definition, a {\em holomorphic differential operator of order $m$} is a linear map
\ben
D : \Gamma^{hol}(X ; E) \to \Gamma^{hol}(X ; F)
\een
such that, with respect to a holomorphic coordinate chart $\{z_i\}$ on $X$, $D$ can be written as
\be\label{local holomorphic}
D|_{\{z_i\}} = \sum_{|I| \leq m} a_I (z) \frac{\partial^{|I|}}{\partial z_I}
\ee
where $a_I(z)$ is a local holomorphic section of ${\rm Hom}(E,F)$.
Here, the sum is over all multi-indices $I = (i_1,\ldots, i_d)$ and 
\ben
\frac{\partial^{|I|}}{\partial z_I} := \prod_{k=1}^d \frac{\partial^{i_k}}{\partial z_k^{i_k}} . 
\een 
The length of the multi-index $I$ is defined by $|I| := i_1 + \cdots + i_d$. 

\begin{eg}
The most basic example of a holomorphic differential operator is the holomorphic de Rham operator $\partial$. 
For each $1 \leq \ell \leq d = \dim_\CC(X)$, it is a holomorphic differential operator from $E = \wedge^\ell T^{1,0*}X$ to $F = \wedge^{\ell+1} T^{1,0*}X$ which on sections is
\ben
\partial : \Omega^{\ell, hol}(X) \to \Omega^{\ell+1, hol}(X) .
\een
Locally, of course, it has the form
\ben
\partial = \sum_{i = 1}^{d} (\d z_i \wedge (-)) \frac{\partial}{\partial z_i},
\een
where $\d z_i \wedge (-)$ is the vector bundle homomorphism $\wedge^\ell T^{1,0*}X \to \wedge^{\ell+1} T^{1,0*}X$ sending $\alpha \mapsto \d z_i \wedge \alpha$. 
\end{eg}

The next piece of data we fix is:
\begin{itemize}
\item[(2)] a square-zero holomorphic differential operator 
\ben
Q^{hol} : \sV^{hol} \to \sV^{hol}
\een
of cohomological degree $+1$. 
Here $\sV^{hol}$ denotes the holomorphic sections of $V$. 
\end{itemize}

Finally, to define a free theory we need the data of a shifted symplectic pairing. 
For reasons to become clear in a moment, we must choose this pairing to have a strange cohomological degree. 
The last piece of data we fix is:
\begin{itemize}
\item[(3)] an invertible bundle map
\ben
(-,-)_V : V \tensor V \to K_X[d-1]
\een
Here, $K_X$ is the canonical bundle on $X$. 
\end{itemize}

The definition of the fields of an ordinary field theory are the {\em smooth} sections of the vector bundle $V$. 
In our situation this is a silly thing to do since we lose all of the data of the complex structure we used to define the objects above.
The more natural thing to do is to take the {\em holomorphic} sections of the vector bundle $V$. 
By construction, the operator $Q^{hol}$ and the pairing $(-,-)_V$ are defined on holomorphic sections, so on the surface this seems reasonable.
The technical caveat that the sheaf of holomorphic sections does not satisfy certain conditions necessary to study renormalization and observables in our approach to QFT. 
For more details on this see Remark \ref{rmk: hol sec bad}.
The solution to this problem is to take a natural resolution of holomorphic sections in order to relate to the usual definition of a classical BV theory.

Given any holomorphic vector bundle $V$ we can define its {\em Dolbeault complex} $\Omega^{0,*}(X , V)$ with its Dolbeault operator 
\ben
\dbar : \Omega^{0,p}(X, V) \to \Omega^{0,p+1}(X, V) .
\een
Here, $\Omega^{0,p}(X, V)$ denotes smooth sections of the vector bundle $\Wedge^p (T^{0,1})^{\vee} X \tensor V$. 
For any $U \subset X$ open subset, the complex $\Omega^{0,*}(U,V)$ is defined. 
In this way, we obtain a natural sheaf of complexes on $X$, that we denote by $\Omega^{0,*}_X(V)$. 
The fundamental property of the Dolbeault complex is that by Dolbeault's Theorem it provides a resolution for the sheaf of holomorphic sections: 
\ben
\sV^{hol} \to \Omega^0_X(V) \xto{\dbar} \Omega^{0,1}_X(V) \xto{\dbar} \cdots .
\een 

We now take a graded holomorphic vector bundle $V = V^{\bullet}$ as above, equipped with the differential operator $Q^{hol}$. 
The Dolbeault resolution $\Omega^{0,*}(X, V^\bullet)$ is now equipped with two differentials $Q^{hol}$ and $\dbar$. 
The complex of fields is the totalization of this complex:
\ben
\sE_V = {\rm Tot}\left(\Omega^{0,*}(X, V), \dbar, Q^{hol}\right) = \left(\Omega^{0,*}(X, V), \dbar + Q^{hol}\right) .
\een
The operator $\dbar + Q^{hol}$ will be the linearized BRST operator of our theory.
By assumption, we have $[\dbar, Q^{hol}] = 0$ so that $(\dbar + Q^{hol})^2 = 0$ and hence the fields still define a complex. 

By construction, $\sE_V$ has the natural structure of a sheaf of complexes.
When we want to consider global sections over $X$ we use the notation $\sE_V(X)$. 
There is similarly a cosheaf of compactly supported sections $\sE_{V,c}$ whose underlying graded is the compactly supported Dolbeault forms $\Omega^{0,*}_c(X, V)$. 

The pairing $(-,-)_V$ defines a pairing on $\sE_V$ as follows.
The thing to observe here is that $(-,-)_V$ extends to the Dolbeault complex in a natural way: we simply combine the wedge product of forms with the pairing on $V$.
We obtain the following composition. 
\ben
\xymatrix{
\sE_{V,c} \tensor \sE_{V,c} \ar[r]^-{(-,-)_V} \ar@{.>}[dr]_-{\omega_V} & \Omega^{0,*}_c(X , K_X) [d-1] \ar[d]^-{\int_X} \\
& \CC[-1] .
}
\een
The top Dolbeault forms with values in the canonical bundle $K_X$ are precisely the top forms on the smooth manifold $X$, and we use the integration map $\int_X : \Omega^{d,d}_c(X) \to \CC$. 
We note that integration is of cohomological degree $d$, as exhibited in the diagram. 

We arrive at the following definition. 

\begin{dfn/lem}\label{dfn hol free theory}
A {\em free holomorphic theory} on a complex manifold $X$ is the data $(V, Q^{hol}, (-,-)_V)$ as in (1), (2), (3) above such that $Q^{hol}$ is a square zero holomorphic differential operator that is graded skew self-adjoint for the pairing $(-,-)_V$.
The triple $(\sE_V, Q_V = \dbar + Q^{hol}, \omega_V)$ defines a free BV theory in the usual sense.
\end{dfn/lem}

The usual prescription for writing down the associated action functional holds in this case.
If $\varphi \in \Omega^{0,*}(X , V)$ denotes a field the action is
\ben
S(\varphi) = \int_X \left(\varphi, (\dbar + Q^{hol}) \varphi \right)_V .
\een

We arrive at an example, which is a higher dimensional version of a familiar chiral CFT. 

\begin{eg}\label{eg bg} {\em The free $\beta\gamma$ system}.
Suppose that 
\ben
V = \ul{\CC} \oplus K_X [d - 1] .
\een
Let $(-,-)_V$ be the pairing
\ben
(\ul{\CC} \oplus K_X) \tensor (\ul{\CC} \oplus K_X) \to K_X \oplus K_X \to K_X 
\een 
sending $(\lambda, \mu) \tensor (\lambda',\mu') \mapsto (\lambda \mu', \lambda'\mu) \mapsto \lambda\mu' + \lambda' \mu$.
In this example we set $Q^{hol} = 0$. 
One immediately checks that this is a holomorphic free theory as above.
The space of fields can be written as
\ben
\sE_V = \Omega^{0,*}(X) \oplus \Omega^{d,*}(X)[d - 1] .
\een 
We write $\gamma \in \Omega^{0,*}(X)$ for a field in the first component, and $\beta \in \Omega^{d,*}(X)[d - 1]$ for a field in the second component. 
The action functional reads
\ben
S(\gamma + \beta) = \int_{X} \beta \wedge \dbar \gamma .
\een 
When $d = 1$ this reduces to the ordinary chiral $\beta\gamma$ system from conformal field theory. 
The $\beta\gamma$ system is a bosonic version of the ghost $bc$ system that appears in the quantization of the bosonic string, see Chapter 6 of \cite{Polchinski1}.
For instance, we will see how this theory is the starting block for constructing general holomorphic $\sigma$-models. 
\end{eg}

Of course, there are many variants of the $\beta\gamma$ system that we can consider.

\begin{eg}{\em Coefficients in a bundle}
For instance, if $E$ is {\em any} holomorphic vector bundle on $X$ we can take 
\ben
V = E \oplus K_{\CC^d} \tensor E^\vee [d-1]
\een
where $E^\vee$ is the linear dual bundle. 
The pairing is constructed as in the case above where we also use the evaluation pairing between $E$ and $E^\vee$.
In thise case, the fields are $\gamma \in \Omega^{0,*}(X, E)$ and $\beta \in \Omega^{d,*}(X, E^\vee)[d-1]$. 
The action functional is simply
\ben
S(\gamma + \beta) = \int {\rm ev}_E(\beta \wedge \dbar \gamma) .
\een 
Here, ${\rm ev}_E$ stands for the evaluation pairing between sections of $E$ and sections of the dual $E^\vee$.
When $E$ is a tensor bundle of type $(r,s)$ this theory is a bosonic version of the $bc$ ghost system of spin $(r,s)$. 
For a general bundle $E$ we will refer to it as the $\beta\gamma$ system with coefficients in the bundle $E$. 
\end{eg}
%In \cite{BWhCDO} we study the quantization of the {\em curved} higher dimensional $\beta\gamma$ system and its relationship to complex invariants generalizing that of the elliptic genus. 

%\begin{eg}
%{\em The free chiral scalar}.
%Another basic example is the free chiral scalar. 
%This is a bit outside\brian{finish}
%Let $X$ be a complex manifold with Hermitian metric $g$. 
%Let $V = \ul{\CC}$, the trivial vector bundle. 
%\brian{do this}
%\end{eg}

\begin{rmk} \label{rmk: hol sec bad}
We will only work with a holomorphic theory prescribed by the data $(V, (-,-)_V, Q^{hol})$ through its associated BV theory.
One might propose a definition of a BV theory in the analytic category based off of holomorphic sections of holomorphic vector bundles. 
There are numerous technical reason why this approach fails in our approach to QFT.
In particular, the sheaf of holomorphic sections of a holomorphic bundle is not fine, and there do not exists partitions of unity in general. 
In addition, there is no holomorphic analog of compactly supported smooth functions. 
Compactly supported functions are imoportant when considering locality in field theory. 
For instance, the main result of \cite{CG2} is that the observables of any QFT form a factorization algebra, which is heavily on the existence of sections with compact support. 
\end{rmk}

\subsection{Interacting holomorphic field theories} \label{sec: interacting}

\def\olochol{\sO_{\rm loc}^{hol}}

We proceed to define what an interacting holomorphic theory is.
A general interacting field theory with space of fields $\sE$ is prescribed by a functional
\ben
S : \sE \to \CC
\een
that satisfies the {\em classical master equation}.
The key technical condition is that this functional must, in addition, be {\em local}.

%Recall, the sheaf of local functionals on $\sE = \Gamma(E)$ is defined as the sheaf of Lagrangian densities
%\ben
%\oloc(\sE) = {\rm Dens}_M \tensor_{D_M} \sO_{red}(JE) .
%\een
%In the expression above $JE$ stands for the sheaf of smooth sections of the $\infty$-jet bundle ${\rm Jet}(E)$ which has the structure of a $D_X$-module.

Since $X$ is a complex manifold, it makes sense to consider the sheaf of holomorphic differential operators that we denote by $D_X^{hol}$. 
If $V$ is a holomorphic vector bundle we define the bundle of holomorphic $\infty$-jets ${\rm Jet}^{hol}(V)$ as follows \cite{GriffithsGreen, WongChandler}. 
This is a pro-vector bundle that is holomorphic in a natural way.
The fibers of this infinite rank bundle ${\rm Jet}^{hol}(V)$ are isomorphic to 
\ben
{\rm Jet}^{hol}(V)|_w = V_w \tensor \CC[[z_1,\ldots,z_d]],
\een
where $w \in X$ and where $\{z_i\}$ is the choice of a holomorphic formal coordinate near $w$. 
We denote by $J^{hol} V$ the sheaf of holomorphic sections of this jet bundle.
The sheaf $J^{hol}V$ has the structure of a $D_X^{hol}$-module, that is, it is equipped with a holomorphic flat connection $\nabla^{hol}$.
This situation is completely analogous to the smooth case.
Locally, the holomorphic flat connection on ${\rm Jet}^{hol}(V)$ is of the form
\ben
\nabla^{hol} |_w = \sum_{i=1}^d \d w_i \left(\frac{\partial}{\partial w_i} - \frac{\partial}{\partial z_i}\right),
\een
where $\{w_i\}$ is the local coordinate on $X$ near $w$ and $z_i$ is the fiber coordinate labeling the holomorphic jet expansion.
%Using holomorphic jets we can make a completely analogous definition in our setting.

One natural appearance of the bundle of holomorphic jets is in providing an explicit description of holomorphic differential operators. 
The statement in the smooth category is simply that a differential operator between vector bundles is equivalent to the data of a map of $D$-modules between the associated $\infty$-jet bundles.
In a completely analogous way, holomorphic differential operators are the same as bundle maps between the associated holomorphic jet bundles. 
A similar result holds for {\em poly}differential operators, which we also state.

\begin{lem}
Suppose $V,W$ are holomorphic vector bundles with spaces of holomorphic sections given by $\sV^{hol},\sW^{hol}$ respectively.
There is an isomorphism of sheaves on $X$
\ben
{\rm Diff}^{hol}(\sV^{hol}, \sW^{hol}) \cong {\rm Hom}_{D_X^{hol}} (J^{hol}(V), J^{hol}(W)) .
\een
Similarly, if $V_1,\ldots,V_n,W$ are holomorphic bundles on $X$, there is an isomorphism
\ben
{\rm PolyDiff}^{hol}(\sV_1^{hol} \times \cdots \times \sV_n^{hol}, \sW^{hol}) \cong {\rm Hom}(J^{\rm hol}(V_1) \tensor \ldots \tensor J^{\rm hol}(V_n), W) .
\een
In both cases, the right-hand side denotes the space of homomorphisms of holomorphic $D$-modules that are compatible with the adic topology on jets.
\end{lem}

We will utilize this intepretation of holomorphic jet bundles momentarily.

In ordinary field theory, local functionals are defined as integrals of Lagrangian densities. 
By definition, a Lagrangian density is a density valued functional on the fields that only depends on the fields through its partial derivatives.
In the holomorphic setting we have the following definition.

\begin{dfn}\label{dfn hol lag}
Let $V$ be a vector bundle.
The sheaf of {\em holomorphic Lagrangian densities} on $V$ is
\ben
{\rm Lag}^{hol}(V) = \Omega^{d,hol}_X \tensor_{\sO^{hol}_X} \left(\prod_{n > 0} {\rm Hom}_{\sO^{hol}_X} (J^{hol}(V)^{\tensor n} , \sO^{hol}_X)_{S_n} \right) .
\een
The Hom-space inside the parentheses denotes maps of holomorphic vector bundles respecting the natural filtration on $\infty$-jets.
That is, we require the bundle maps to be continuous with respect to the natural adic topology.
We also take coinvariants for the symmetric group $S_n$.
\end{dfn}

Note that we take the product over $n > 0$ so we do not want to consider Lagrangians that are constant in the fields.

Equivalently, a holomorphic Lagrangian density is of the form $\omega \tensor F$ where $\omega$ is a top holomorphic form and $F$ is a functional $F = \sum_k F_k$ where, for each $k$, the multilinear map
\ben
F_k : \sV^{hol} \times \cdots \times \sV^{hol} \to \sO^{hol}_X
\een
depends only on the holomorphic $\infty$-jet of sections of $V$. 

%\begin{dfn}
%Suppose $V$ is a graded holomorphic vector bundle.
%We define the sheaf of {\em holomorphic} local functionals on $V$ by
%\ben
%\olochol(V) = \Omega^{d,hol}_X \tensor_{D^{hol}_X} \sO_{red}(J^{hol}V) [d]
%\een
%\end{dfn}

The next definition we will need is that of a holomorphic local functional. 
Just as in Definition \ref{dfn: local fnl}, this is given by the sheaf of Lagrangians modulo total derivatives. 
Of course, in this setting we require both the Lagrangians and derivatives to be holomorphic in the appropriate sense. 

\begin{dfn}
\label{dfn: interacting}
The sheaf of {\em holomorphic local functionals} is defined to be the quotient
\be\label{quotient}
\olochol(V) := {\rm Lag}^{hol}(V) / \sT_X^{hol} \cdot {\rm Lag}^{hol}(V),
\ee
where $\sT_X \cdot {\rm Lag}^{hol}(V)$ denotes the subsheaf of holomorphic Lagrangians that are in the image of the Lie derivative by some holomorphic vector field.
Given a holomorphic Lagrangian $I^{hol} \in {\rm Lag}^{hol}(V)$, we denote its class in local functionals by $\int I^{hol} \in \olochol(V)$. 
\end{dfn}

Equivalently, we may express the quotient (\ref{quotient}) using holomorphic $D$-modules in the following way.
The left $D_{X}^{hol}$-module structure on $J^{hol}(V)$ carries over to a left $D_X^{hol}$-module structure on the product 
\ben
\prod_{n > 0} {\rm Hom}_{\sO^{hol}_X} (J^{hol}(V)^{\tensor n}, \sO^{hol}_X) .
\een
Using the natural structure of a right $D_X^{hol}$-module structure on $\Omega^{d,hol}_X$, we obtain an isomorphism
\ben
\olochol(V) = \Omega^{d,hol}_X \tensor_{D_X^{hol}} \left( \prod_{n > 0} {\rm Hom}_{\sO^{hol}_X} (J^{hol}(V)^{\tensor n}, \sO^{hol}_X)_{S_n}\right) .
\een

We use the notion of a holomorphic local functional to formulate the definition of an interacting holomorphic field theory.
Suppose that $V$ is part of the data of a free holomorphic theory $(V, Q^{hol},(-,-)_V)$.
The pairing $(-,-)_V$ endows the space of holomorphic local functionals with a bracket on local functionals. 
Likewise, the operator $Q^{hol}$ determines a differential on local functionals.
These facts are summarized in the following lemma. 

\begin{lem}
Suppose $(V, Q^{hol},(-,-)_V)$ is the data of a free holomorphic theory. 
The pairing $(-,-)_V$ (respectively, operator $Q^{hol}$) defines a bracket $\{-,-\}^{hol}$ (respectively, differential $Q^{hol}$) on $\olochol(V)$ of degree $d-1$ (respectively, degree $+1$).

This yields the structure of a sheaf of dg Lie algebras:
\ben
\left(\olochol(V)[d-1], Q^{hol}, \{-,-\}^{hol}\right)
\een
where $Q^{hol}$ is the differential, and $\{-,-\}^{hol}$ is the Lie bracket. 
\end{lem}
\begin{proof}
The operator $Q^{hol}$ extends to an operator on holomorphic Lagrangians in the obvious way. 
Indeed, if a holomorphic Lagrangian is of the form $I^{hol} = \omega \tensor F$ where $F : \sV^{hol} \times \sV^{hol} \to \sO^{hol}$, then we define $Q I^{hol} = \omega \tensor (Q^{hol} F)$.
This descends to an operator on $\olochol(V)$ of cohomological degree $+1$. 

Next, we show how $(-,-)_V$ defines a bracket $\{-,-\}^{hol}$ on the sheaf of holomorphic local functionals. 
Denote by $\omega^{hol} : \sV^{hol} \tensor_{\sO^{hol}} \sV^{hol} \to \Omega^{d,hol}$ the map of $\sO^{hol}$-modules which sends $\varphi_1 \tensor \varphi_2 \mapsto (\varphi_1, \varphi_2)_V$. 
A derivation $X$ of the algebra ${\rm Hom}_{\sO^{hol}} (\sV^{hol}, \sO^{hol})$ can be expanded in components of the form
\[
X^{(k)} : (\sV^{hol})^{\tensor k} \to \sV^{hol} .
\]
Using $\omega^{hol}$, we define the $S_{k+1}$-invariant map of $\sO^{hol}_X$-modules
\[
X^{(k)} \vee \omega^{hol} : (\sV^{hol})^{\tensor (k+1)} \to \Omega^{d,hol}
\]
which sends $\varphi_1 \tensor \cdots \tensor \varphi_{k+1} \mapsto \omega(\varphi_1, X^{(k)} (\varphi_2, \ldots, \varphi_{k+1}))$. 

Any holomorphic local functional $\int I^{hol} \in \olochol(V)$ is Hamiltonian in the sense that there exists a derivation
$X_I^{hol}$ such that $\int I^{hol} = \sum_{k \geq 0} X^{(k),hol}_I \vee \omega^{hol}$. 
Moreover, $X_I^{hol}$ is a local derivation in the sense that it preserves the subspace of local holomorphic functionals. 

Given $\int I^{hol} , \int J^{hol} \in \olochol(V)$, we define 
\[
\left\{\int I^{hol} , \int J^{hol} \right\}^{hol} = X_I^{hol} \left(\int J^{hol} \right) \in \olochol(V) .
\]
Note that if $\int I^{hol}$ is of total cohomological degree $\ell$, then the derivation $X_{I}^{hol}$ is of degree $\ell + d - 1$. 
Thus, the bracket $\{-,-\}^{hol}$ is of degree $-d+1$. 
It is immediate to verify that this bracket satisfies the appropriate graded Jacobi identity and that $Q^{hol}$ acts as a graded derivation. 

%We can assume, without loss of generality, that $\int I^{hol}$ and $\int J^{hol}$ are of polynomial degree $k$ and $\ell$, respectively.
%This means that 
%\begin{align*}
%\int I^{hol} & \in \Omega^{d,hol}_X \tensor_{D_X^{hol}} {\rm Hom}_{\sO^{hol}_X} (J^{hol}(V)^{\tensor k}, \sO^{hol}_X)_{S_n}\\
%\int J^{hol} & \in \Omega^{d,hol}_X \tensor_{D_X^{hol}} {\rm Hom}_{\sO^{hol}_X} (J^{hol}(V)^{\tensor \ell}, \sO^{hol}_X)_{S_n} . 
%\end{align*}
%For simplicity, for the explicit dependence on the fields $\varphi_1,\ldots, \varphi_k \in \sV^{hol}$ we take $\int I^{hol}$ to be
%\[
%\int I^{hol} = \int D_1 (\varphi_1) \cdots D_k (\varphi_k) \omega 
%\]
%where $D_1,\ldots, D_k : \sV^{hol} \to \sO^{hol}_X$ are holomorphic differential operators and $\omega \in \Omega^{d,hol}_X$.
%For the dependence on the fields $\varphi_{k+1}, \ldots, \varphi_{k + \ell}$, we take $\int J^{hol}$ to be
%\[
%\int J^{hol} = \int D_{k+1} (\varphi_{k+1}) \cdots D_{\ell + k} (\varphi_{\ell+k})  \eta
%\]
%where $D_{k+1},\ldots, D_{\ell+k} : \sV^{hol} \to \sO_X^{hol}$ are holomorphic differential operators and $\eta \in \Omega^{d,hol}_X$. 
%
%To 
%\brian{fix all in comments into a proof}

\end{proof}

We can now state the definition of a classical holomorphic theory. 
The definition involves a holomorphic Lagrangian $I^{hol}$ that is at least {\em cubic}.
For brevity, we will make the following definition.

\begin{dfn}\label{dfn: plus}
The subsheaf of cubic holomorphic Lagrangians is
\[
{\rm Lag}^{hol,+}(V) := \Omega^{d,hol}_X \tensor_{\sO^{hol}_X} \prod_{n \geq 3} {\rm Hom} ({\rm Jet}^{hol}(V)^{\tensor n} , K_X)_{S_n}  \subset {\rm Lag}^{hol}(V)
\]
and the corresponding space of local functionals will be denoted $\sO_{\rm loc}^{hol,+}(V)$.
\end{dfn}

%Note that top holomorphic forms have a natural action by the Lie algebra of holomorphic vector fields $\sT_X^{hol}$ via Lie derivative.
%This induces an action of holomorphic vector fields on the space of holomorphic Lagrangians.

% is a homogenous holomorphic Lagrangian then $L_\xi \cdot F$ is the holomorphic Lagrangian defined by
%\ben
%(\varphi_1,\ldots,\varphi_n) \mapsto L_\xi F(\varphi_1,\ldots,\varphi_n) .
%\een

\begin{dfn}
A {\em classical holomorphic theory} on a complex manifold $X$ is the data of a free holomorphic theory $(V, Q^{hol}, (-,-)_V)$ plus a holomorphic Lagrangian
\ben
I^{hol} \in {\rm Lag}^{hol,+}(V)
\een
of cohomological degree $d$, such that the local functional $\int I^{hol} \in \sO_{\rm loc}^{hol,+}(V)$ is a solution to the Maurer-Cartan equation in the dg Lie algebra $(\olochol(V)[d-1], Q^{hol}, \{-,-\}^{hol})$:
\ben
Q^{hol} \int I^{hol} + \frac{1}{2} \left\{\int I^{hol}, \int I^{hol}\right\}^{hol} = 0 .
\een 
%is in the image of a some holomorphic vector field.
%That is, there exists some $\xi \in \sT_X^{hol}$ and $F \in {\rm Lag}^{hol}(V)$ such that $Q^{hol}I^{hol} + \frac{1}{2} \{I^{hol}, I^{hol}\}^{hol} = L_\xi F$.
\end{dfn}

%There is an alternative way to understand the condition on the functional $I^{hol}$ to define a classical holomorphic theory.
%To define it, we introduce the notion of a holomorphic local functional.
%A holomorphic local functional is, by definition, a holomorphic Lagrangian defined up to a total holomorphic derivative.
%Precisely, we have the following definition.
%
%\brian{fix}
%
%Compare this to the definition of ordinary local functionals that we recalled in Definition \ref{dfn: local fnl}.
%
%
%An immediate corollary of this lemma is that the condition for a holomorphic Lagrangian $I^{hol}$ to define a classical theory in Definition \ref{dfn: interacting} is equivalent to the condition that $I^{hol}$ is a Maurer-Cartan element in $\olochol(V)[d-1]$.

%\begin{dfn/lem}
%Let $(V, Q^{hol}, (-,-)_V, I^{hol})$ be the data of an interacting holomorphic theory. 
%Then $Q^{hol} + \{I^{hol},-\}$ equips $\olochol(V)$ with the structure of a sheaf of cochain complexes that we will denote
%\ben
%\Def^{hol}_{V} := \left(\olochol(V), Q^{hol} + \{I^{hol}, -\}^{hol}\right) .
%\een
%\end{dfn/lem}

As in the free case, we proceed to verify that a holomorphic theory defines an interacting classical BV theory in the sense of Definition \ref{dfn: classical}.

The underlying space of fields, as we have already seen in the free case, is $\sE_V = \Omega^{0,*}(X , V)$. 
We show how to extend a holomorphic Lagrangian to a functional on this Dolbeualt complex.

Recall, a holomorphic Lagrangian can be written as $I^{hol} = \sum_k I^{hol}_k$ where $I_k^{hol} = \omega \tensor F_k$ for $\omega \in \Omega^{d,hol}$ and $F_k : \sV^{hol} \times \cdots \times \sV^{hol} \to \sO^{hol}$ is of the form 
\ben
F_k (\varphi_1,\ldots,\varphi_k) = \sum_{i_1,\ldots,i_k} D_{i_1}(\varphi_1)\cdots D_{i_k}(\varphi_k) \in \sO^{hol}_X .
\een
Here, $\varphi_i\in \sV^{hol}$ is a holomorphic section, and each $D_{i_j}$ is a holomorphic differential operator $D_{i_j} : \sV^{hol} \to \sO^{hol}$.

In general, suppose $V,W$ are holomorphic vector bundles.
Every holomorphic differential operator $D : \sV \to \sW$ extends to a {\em smooth} differential operator on the associated Dolbeualt complexes with the property that it is compatible with the $\dbar$-operator on both sides.

To see how this works, suppose $D : \sV^{hol} \to \sW^{hol}$ is locally of the form $$D = \sum_{m_1,\ldots,m_d} a_{m_1\cdots m_d}(z)\frac{\partial^{m_1}}{\partial z_1^{m_1}} \cdots \frac{\partial^{m_d}}{\partial z^{m_d}}$$ where $a_{m_1\cdots m_d}(z)$ denotes a local holomorphic section of ${\rm Hom}(V,W)$. 
Then, if $\alpha = s_I (z,\zbar) \d \zbar_I \in \Omega^{0,*}(X, V)$, where $s_I$ is a local {\em smooth} section of $V$, we define
\ben
D^{\Omega^{0,*}} \alpha = \sum_{m_1,\ldots,m_d} a_{m_1\cdots m_d} (z)\left(\frac{\partial^{m_1}}{\partial z_1^{m_1}} \cdots \frac{\partial^{m_d}}{\partial z^{m_d}} s_I(z,\zbar)\right) \d \zbar_I \in \Omega^{0,*}(X, W) .
\een
In this way, $D$ extends to a differential operator 
\ben
D^{\Omega^{0,*}} : \Omega^{0,*}(X, V) \to \Omega^{0,*}(X,W) . 
\een
Since $D$ is holomorphic, it is immediate that $D^{\Omega^{0,*}} \dbar_V = \dbar_W D^{\Omega^{0,*}}$ where $\dbar_V,\dbar_W$ are the $(0,1)$-connections on $V,W$ respectively.
Thus, $D^{\Omega^{0,*}}$ is a map of sheaves of cochain complexes.

Via this construction, we extend $F_k$ to a $\Omega^{0,*}(X)$-valued functional on $\Omega^{0,*}(X, V)$ by the formula
\ben
F^{\Omega^{0,*}}_k : (\alpha_1,\ldots,\alpha_k) \mapsto \sum_{i_1,\ldots,i_k} D^{\Omega^{0,*}}_{i_1}(\alpha_1) \wedge \cdots \wedge D^{\Omega^{0,*}}_{i_k} (\alpha_k) \in  \Omega^{0,*}(X) .
\een
Here, as above, the $\alpha_i$'s denote sections in $\Omega^{0,*}(X, V)$. 

We have thus produced a linear map
\ben
(-)^{\Omega^{0,*}} : {\rm Hom}_{\sO^{hol}} ((J^{hol} V)^{\tensor k} , \sO^{hol}) \to {\rm Hom}_{C^\infty} ((J \Omega^{0,*}(X, V))^{\tensor k} , \Omega^{0,*}(X))
\een
where $J \Omega^{0,*}(X, V)$ denotes the sheaf of smooth jets of the graded vector bundle underlying the Dolbeault complex.
This map clearly restricts to the symmetric coinvariants on both sides.
Taking direct products and tensoring with $\Omega^{d,hol}_X$ we have a map 
\ben
{\rm Lag}^{hol}(V) \to \Omega^{d,hol} \tensor_{C^\infty} \prod_{k > 0} {\rm Hom}_{C^\infty}(J \Omega^{0,*}(X,V)^{\tensor k} , \Omega^{0,*}(X)) \cong \Omega^{d,*} \tensor_{C^\infty} \prod_{k > 0} {\rm Hom}(J \Omega^{0,*}(X, V)^{\tensor k} , C^\infty) .
\een
We have already mentioned that this map is compatible with the $\dbar$-operator on the right-hand side.
Moreover, the holomorphic differential operator $Q^{hol}$ also extends to a differential operator on the right-hand side in a way compatible with $\dbar$.
Thus, $(-)^{\Omega^{0,*}}$ is a map of cochain complexes, where ${\rm Lag}^{hol}(X)$ is equipped with the differential $Q^{hol}$ and the right-hand side has differential $\dbar + Q^{hol}$. 

The right-hand side admits a map of degree $-d$ to $\Omega^{d,d} \tensor_{C^\infty} \prod_{k > 0} {\rm Hom}(J \Omega^{0,*}(X, V)^{\tensor k} , C^\infty)$ by projecting onto the $(d,d)$-component of $\Omega^{d,*}$. 
Note that this map is only graded linear, it does not preserve the $\dbar$-differential. 
However, once we quotient by the action of vector fields we do get a well-defined map
\ben
\olochol(V) \to \Omega^{d,d} \tensor_{D} \prod_{k > 0} {\rm Hom}(J \Omega^{0,*}(X, V)^{\tensor k} , C^\infty)_{S_k} [-d] .
\een
Note that we have accounted for the shift of $d$ coming from $\Omega^{d,*} \to \Omega^{d,d}[-d]$. 
The right-hand side is precisely the (shifted) space of ordinary local functionals for the sheaf $\sE_V = \Omega^{0,*}(X,V)$ defined in Definition \ref{dfn: local fnl}.

In conclusion, we have obtained the following map of sheaves of cochain complexes
\be\label{int eqn}
\int (-)^{\Omega^{0,*}} : \olochol(V) \to \oloc(\Omega^{0,*}(X,V))[-d] .
\ee
In fact, we have the following stronger result, that this map is compatible with the brackets on both sides.

\begin{lem}
The map $\int (-)^{\Omega^{0,*}}$ defines an map of sheaves of dg Lie algebras
\ben
\int(-)^{\Omega^{0,*}} : \olochol(V)[d-1] \to \oloc(\sE_V)[-1]
\een
\end{lem}

\begin{proof}
By definition, the sheaf of local functionals on $\sE_V$ is equal to 
\ben
{\rm Dens}_X \tensor_{D_X} \sO_{red}(J \sE_V) .
\een
Since $\sO_{red}(J \sE_V)$ is flat as a $D_X$-module \cite{CosRenorm}, we can replace the tensor product $\tensor_{D_X}$ with the derived tensor product $\tensor^{\LL}_{D_X}$.

We now use the following observation about $D$-modules.
If $M$ is a holomorphic $D_{X}^{hol}$-module which is given by the sections of a holomorphic vector bundle, then it forgets down to an ordinary smooth $D_X$-module (with the same underlying $C^\infty_X$-module structure) that we denote $M^{C^\infty}$. 
Moreover, there is a quasi-isomorphism of $D$-modules
\ben
\Omega^{d,hol}_X \tensor^{\LL}_{D_X^{hol}} M [d] \simeq \Omega^{d,d}_{X} \tensor^{\LL}_{D_X} M^{C^\infty} .
\een

We apply this to the case $M = \sO_{red}(J^{hol} V)$, where $V$ is a holomorphic vector bundle
This says that there is a quasi-isomorphism
\be\label{def eqn1}
\Omega^{d,hol}_X \tensor^{\LL}_{D_X^{hol}} \sO_{red}(J^{hol} V) [d] \simeq \Omega^{d,d}_{X} \tensor^{\LL}_{D_X}  \sO_{red}(J^{hol} V) .
\ee
This quasi-isomorphism is compatible with the $Q^{hol}$ differential and the bracket $\{-,-\}^{hol}$ on both sides.
Note that the left-hand side is simply the space of shifted holomorphic local functionals $\olochol(V)[d]$. 

Next, observe that the map $(-)^{\Omega^{0,*}}$ determines a map of sheaves of cochain complexes 
\be\label{def eqn2}
(-)^{\Omega^{0,*}} : \Omega^{d,d}_{X} \tensor^{\LL}_{D_X}  \sO_{red}(J^{hol} V) \to {\rm Dens}_X \tensor^{\LL}_{D_X} \sO_{red}(J \sE_V) .
\ee
The right-hand side is quasi-isomorphic to $\oloc(\sE_V)$. 
The composition of (\ref{def eqn1}) and (\ref{def eqn2}) is simply the map (\ref{int eqn})
\ben
\int (-)^{\Omega^{0,*}} : \olochol(V)[d] \to \oloc(\sE_V) .
\een
One checks immediately that this map is compatible with the brackets, namely 
\ben
\{\int I^{hol}, \int J^{hol}\}^{hol} = \{\int I^{\Omega^{0,*}}, \int J^{\Omega^{0,*}}\} .
\een

%$D_X$-modules between $\sO_{red}(J \Omega^{0,*}(X , V))$ and $\sO_{red}(J^{hol}V)$. 
%For this, it suffices to show that the space of linear functionals are quasi-isomorphic. 
%For any vector bundle $E$ there always exists a (non-canonical) splitting $J E \cong \sE \tensor_{C^\infty_X} J_X$, where $\sE$ is the sheaf of sections and $J_X$ is the sheaf of $\infty$-jets of the trivial bundle.
%Thus, we can assume that $V$ is the trivial vector bundle, where the claim is now $(J\Omega^{0,*}(X))^\vee \simeq (J^{hol}_X)^\vee$. 
%Both sides are quasi-isomorphic to the smooth sections of the bundle of holomorphic differential operators $D^{hol}$, so we are done. 
\end{proof}

As a result of the equivalence between solutions to the classical master equation and Maurer-Cartan elements in the dg Lie algebras of shifted local functionals, we have the following. 

\begin{prop} Every classical holomorphic theory $(V, Q^{hol},(-,-)_V, I^{hol})$ determines the structure of a classical BV theory.
The underlying free BV theory is given in Definition/Lemma \ref{dfn hol free theory} $(\sE_V, Q, \omega_V)$ and the interaction is $I = \int I^{\Omega^{0,*}}$. 
\end{prop}

%\begin{proof}
%We must show that $Q^{hol}I^{hol} + \frac{1}{2} \{I^{hol},I^{hol}\}^{hol}$ exact implies the ordinary classical master equation for $I$:
%\ben
%\dbar I + Q^{hol}I + \frac{1}{2} \{I,I\} = 0 .
%\een
%Since $I^{\Omega^{0,*}}$ is defined using holomorphic differential operators, we see that $(\dbar I^{\Omega^{0,*}} ) (\alpha) = \dbar (I^{\Omega^{0,*}} (\alpha))$. 
%Thus, upon integration, we see that the first term is identically zero.
%
%Now, since $Q^{hol}I^{hol} + \frac{1}{2} \{I^{hol},I^{hol}\}^{hol}$ is $\partial$-exact, $Q^{hol}I^{\Omega^{0,*}} + \frac{1}{2} \{I^{\Omega^{0,*}},I^{\Omega^{0,*}}\}$ is also $\partial$-exact.
%By integration we obtain the result.
%\end{proof}

%\brian{This is now repetitive
%Furthermore, via integration we can define the $\CC$-valued functional on compactly supported sections
%\ben
%I_k = \int_X I^{0,*}_k : \Om ega_c^{0,*}(X, V)^{\tensor k} \to \CC .
%\een
%Performing this for each homogenous piece, we obtain the functional $I = \sum_k \int I^{\Omega^{0,*}}_k$. 
%}
%\brian{This too
%The symbol $\int_X$ reminds us that we are working modulo total derivatives.
%Moreover, since we began with a functional involving differential operators, we see that the above expression defines an element of $\oloc(\sE_V)$. 
%In conclusion, our construction thus outlined determines a linear map $\olochol(V) \to \oloc(\sE_V)$ obtained via the composition $I^{hol} \mapsto I^{\Omega^{0,*}} \mapsto I = \int I^{\Omega^{0,*}}$. 
%Note that since $I^{hol}$ is cohomological degree $d$, the local functional $I^{\Omega^{0,*}}$ is degree zero since integration is degree $-d$. 
%}

Table \ref{table: holtoBV} is a useful summary showing how we are producing a BV theory from a holomorphic theory.

\begin{table}
\begin{center}
\begin{tabular}{ |c|c|c| } 
 \hline
 Holomorphic theory & BV theory \\
 \hline \hline
Holomorphic bundle $V$ & Space of fields $\sE_V = \Omega^{0,*}(X, V)$  \\ 
Holomorphic differential operator $Q^{hol}$ & Linear BRST operator $\dbar + Q^{hol}$ \\ 
Non-degenerate pairing $(-,-)_V$ & $(-1)$-symplectic structure $\omega_{V}$ \\ 
Holomorphic Lagrangian $I^{hol}$ & Local functional $I = \int I^{\Omega^{0,*}} \in \oloc(\sE_V)$ \\ 
 \hline
\end{tabular}
\caption{From holomorphic to BV}
\label{table: holtoBV}
\end{center}
\end{table}

\begin{eg} {\em Holomorphic $BF$-theory} \label{eg: bf}
Let $\fg$ be a Lie algebra and $X$ any complex manifold.
Consider the following holomorphic vector bundle on $X$:
\ben
V = \ul{\fg}_X [1] \oplus K_X \tensor \ul{\fg}_X^\vee [d-2] .
\een
The notation $\ul{\fg}_X$ denotes the trivial bundle with fiber $\fg$. 
The pairing $V \tensor V \to K_X[d-1]$ is similar to the pairing for the $\beta\gamma$ system, except we use the evaluation pairing $\<-.-\>_\fg$ between $\fg$ and its dual $\fg^\vee$. 
In this example, $Q^{hol} = 0$.

We describe the holomorphic Lagrangian.
If $f_i : X \to \CC, i=1,2$ are holomorphic functions and $\beta \in K_X$, consider the trilinear functional
\ben
I^{hol} (f \tensor X + \beta \tensor X^\vee) = f^2 \beta \<X^\vee, [X,X]\>_\fg .
\een
This defines an element $I^{hol} \in \sO_{\rm loc}^{hol,+}(V)$ of degree $d$ and the Jacobi identity for $\fg$ guarantees $\{I^{hol}, I^{hol}\}^{hol} = 0$. 
The fields of the corresponding BV theory are
\ben
\sE_V = \Omega^{0,*}(X, \fg)[1] \oplus \Omega^{d,*}(X, \fg^*) [d-2] .
\een
The induced local functional $I^{\Omega^{0,*}}$ on $\sE_V$ is
\ben
I^{\Omega^{0,*}} (\alpha, \beta) = \int_X \<\beta, [\alpha,\alpha]\>_\fg .
\een
The total action is $S(\alpha,\beta) = \int \<\beta, \dbar \alpha\> + \<\beta,[\alpha,\alpha]\>_\fg$.
This is formally similar to $BF$ theory (see below) and for that reason we refer to it as {\em holomorphic} BF theory. 
The moduli problem this describes is the cotangent theory to the moduli space of holomorphic connections on the trivial $G$-bundle near the trivial bundle.
There is an obvious enhancement that works near any holomorphic principal bundle.
When $d = 2$, in \cite{johansen1}, or for a more mathematical treatment see \cite{CostelloYangian}, it is shown that this theory is a twist of $\cN=1$ supersymmetric pure Yang-Mills on $\RR^4$.
\end{eg}

\begin{eg} {\em Topological $BF$-theory}\label{eg: bftop}
This is a deformation of the previous example that has appeared throughout the physics literature.
Suppose we take as our graded holomorphic vector bundle 
\ben
V = \left(\ul{\fg}_X \tensor \left(\oplus_{k = 0}^d \wedge^k T^{*1,0}X [1-k]\right)\right) \oplus \left(\ul{\fg^*}_X \tensor \left(\oplus_{k = 0}^d \wedge^k T^{*1,0}X [2(d-1)-k]\right)\right) .
\een
Here $\wedge^0 T^{*1,0}X$ is understood as the trivial bundle $\ul{\CC}_X$. 
The pairing is given by combining the evaluation pairing between $\fg$ and $\fg^*$ and taking the wedge product and projecting onto the components isomorphic to $K_X$.
Explicitly, the pairing is equal to the sum of bundle maps of the form
\ben
\ev_{\fg} \tensor \wedge : \left(\ul{\fg}_X \tensor \wedge^k T^{*1,0}X [1-k]\right) \tensor \left(\ul{\fg^*}_X \tensor \wedge^{d-k} T^{*1,0}X [d-1+k]\right) \to K_X [d-1] .
\een
The holomorphic differential is of the form 
\ben
Q^{hol} = {\rm id}_\fg \tensor \partial + {\rm id}_{\fg^*} \tensor \partial,
\een
where $\partial$ is the holomorphic de Rham differential.
The holomorphic interaction is given by combining the Lie algebra structure on $\fg$ with the wedge product of the holomorphic bundles $\wedge^k T^{*1,0}X$. 
We observe that the associated BV theory has classical space of fields given by
\ben
(A,B) \in \sE_V = \Omega^*(X, \fg[1] \oplus \fg^*[2d-2]) 
\een
where $\Omega^*$ is now the {\em full} de Rham complex.
The action functional is
\ben
S = \int_X \<B, \d A\>_\fg + \frac{1}{3} \<B, [A,A]\>_\fg .
\een
As above, $\<-,-\>_\fg$ denotes the pairing between $\fg$ and its dual.
This is the well-known topological BF theory on the even dimensional {\em real} manifold $X$ (of real dimension $2d$). 
It might seem silly that we have used the formalism of holomorphic field theory to describe a very simple topological theory.
We will discuss advantages of this approach at the send of the next section.
In particular, the theory of regularization for holomorphic theories we will employ has peculiar consequences for renormalizing certain classes of topological theories such as topological BF theory.
\end{eg}

%When constructing a BV theory from a holomorphic theory $V \rightsquigarrow \sE_V$ it is natural to study deformations of the theory by holomorphic local functionals. 
%When such a local holomorphic functional $I^{hol}$ is fixed, we denote by ${\rm Def}^{hol}_{V,I^{hol}}$ the cochain complex $\left(\olochol(V), Q^{hol} + \{I^{hol}, -\}\right)$
%We will see in Lemma \ref{dfn: holdef} that this cochain complex controls holomorphic deformations of the classical theory defined by $I = \int I^{\Omega^{0,*}}$. 
%
%\begin{lem}\label{lem: holdef}
%Suppose $(V, Q^{hol}, (-,-)_V, I^{hol})$ is the data of a holomorphic theory, and let $(\sE_V, Q = \dbar, \omega_V, I)$ be the corresponding BV theory.
%Then, there is an inclusion of sheaves of cochain complexes
%\ben
%{\rm Def}^{hol}_{V} [d] \hookrightarrow {\rm Def}_{\sE_V} .
%\een
%that is compatible with the brackets and $\{-,-\}^{hol}$ and $\{-,-\}$ on both sides.
%%\ben
%%\Def_{\sE_V}  \simeq \Omega^{d,hol}_X \tensor^{\LL}_{D_X^{hol}} \sO_{red}(J^{hol}V) [d]
%%\een
%\end{lem}

\begin{rmk}
We have mentioned an alternative formulation of classical field theory in terms of sheaves of $L_\infty$ algebras.
Just as in the ordinary case we can formulate the data of a classical holomorphic theory in terms of sheaves of $L_\infty$ algebras. 
We will not do that here, but hope the idea of how to do so is clear.
\end{rmk}

\begin{rmk}
Our definition of a holomorphic theory is compatible with the definition of a two-dimensional chiral conformal field theory given in \cite{LiVertex} when the complex dimension is $d=1$.
\end{rmk}

\subsection{Holomorphically translation invariant theories} \label{sec: hol trans}

When working on affine space one can ask for a theory to be invariant with respect to translations. 
In this section, we take a break from holomorphic theories defined on general complex manifolds to consider the affine manifold $\CC^d = \RR^{2d}$.
We recall what a {\em holomorphically translation invariant} theory is, and state a general result about deformations for such theories. 
This particular class of theories has been discussed in Chapter 10 of \cite{CG2}, and it is a special case of a general holomorphic theory as defined above.  

%The definition we give of a holomorphically translation invariant %theory is slightly different than the structure of the last %section. 

Let $V$ be a holomorphic vector bundle on $\CC^n$ and suppose we fix an identification of bundles 
\ben
V \cong \CC^d \times V_0
\een
where $V_0$ is the fiber of $V$ at $0 \in \CC^d$. 
We want to consider a classical theory with space of fields given by $\Omega^{0,*}(\CC^d, V) \cong \Omega^{0,*}(\CC^d) \tensor_\CC V_0$. 
Moreover, we want this theory to be invariant with respect to the group of translations on $\CC^d$. 
Per usual, it is best to work with the corresponding Lie algebra of translations. 
Using the complex structure, we choose a presentation for the complex Lie algebra of translations given by
\ben
\CC^{2d} \cong {\rm span}_\CC \left\{\frac{\partial}{\partial z_i}, \frac{\partial}{\partial \zbar_i}\right\}_{1 \leq i \leq d}.
\een

To define a theory, we need to fix a non-degenerate pairing on $V$.
Moreover, we want this to be translation invariant. 
So, suppose
\be\label{pairing 1}
(-,-)_V : V \tensor V \to K_{\CC^d} [d-1]
\ee
is a skew-symmetric bundle map that is equivariant for the Lie algebra of translations. 
The shift is so that the resulting pairing on the Dolbeault complex is of the appropriate degree.
Here, equivariance means that for sections $v,v'$ we have
\ben
\left(\frac{\partial}{\partial z_i} v, v'\right)_V + \left(v, \frac{\partial}{\partial z_i} v'\right)_V  = L_{\partial_{z_i}} (v,v')_V
\een
where the right-hand side denotes the Lie derivative applied to $(v,v')_V \in \Omega^{d,hol}_{\CC^d}$. 
There is a similar relation for the anti-holomorphic derivatives. 
We obtain a $\CC$-valued pairing on $\Omega^{0,*}_c(\CC^d , V)$ via integration:
\be\label{trans pairing}
\int_{\CC^d} \circ (-,-)_V : \Omega^{0,*}_c (\CC^d , V) \tensor \Omega^{0,*}_c(\CC^d , V) \xto{\wedge \cdot (-,-)_V} \Omega^{d,*}(\CC^d) \xto{\int} \CC .
\ee
The first arrow is the wedge product of forms combined with the pairing on $V$. 
The second arrow is only nonzero on forms of type $\Omega^{d,d}$. 
Clearly, integration is translation invariant, so that the composition is as well. 

The pairing (\ref{trans pairing}) together with the differential $\dbar$ are enough to define a free theory. 
However, it is convenient to consider a slightly generalized version of this situation. 
We want to allow deformations of the differential $\dbar$ on Dolbeault forms of the form
\ben
Q = \dbar + Q^{hol}
\een
where $Q^{hol}$ is a holomorphic differential operator of the form
\be\label{hol operator}
Q^{hol} = \sum_I \frac{\partial}{\partial z^I} \mu_I
\ee
where $I$ is some multi-index and $\mu_I : V \to V$ is a linear map of cohomological degree $+1$. 
Note that we have automatically written $Q^{hol}$ in a way that it is translation invariant.
Of course, for this differential to define a free theory there needs to be some compatibility with the pairing on $V$. 

We can summarize this in the following definition, which should be viewed as a slight modification of a free theory to this translation invariant holomorphic setting. 

\begin{dfn} A {\em holomorphically translation invariant free BV theory} is the data of a holomorphic vector bundle $V$ together with
\begin{enumerate}
\item an identification $V \cong \CC^d \times V_0$;
\item a translation invariant skew-symmetric pairing  $(-,-)_V$ as in (\ref{pairing 1});
\item a holomorphic differential operator $Q^{hol}$ as in (\ref{hol operator});
\end{enumerate}
such that the following conditions hold
\begin{enumerate}
\item the induced $\CC$-valued pairing $\int \circ (-,-)_V$ is non-degenerate;
\item the operator $Q^{hol}$ satisfies $(\dbar + Q^{hol})^2 = 0$ and is skew self-adjoint for the pairing:
\ben
\int (Q^{hol} v, v')_V = \pm \int (v, Q^{hol} v').
\een
\end{enumerate}
\end{dfn}

The first condition is required so that we obtain an actual $(-1)$-shifted symplectic structure on $\Omega^{0,*}(\CC^d, V)$. 
The second condition implies that the derivation $Q = \dbar + Q^{hol}$ defines a cochain complex
\ben
\sE_V = \left(\Omega^{0,*}(\CC^d, V), \dbar + Q^{hol}\right),
\een
and that $Q$ is skew self-adjoint for the symplectic structure. 
Thus, in particular, $\sE_V$ together with the pairing define a free BV theory in the ordinary sense. 
In the usual way, we obtain the action functional via
\ben
S(\varphi) = \int (\varphi, (\dbar + Q^{hol}) \varphi)_V .
\een 

Before going further, we will give a familiar example from the last section.

\begin{eg}\label{eg bg affine} {\em The free $\beta\gamma$ system on $\CC^d$}.
Consider the $\beta\gamma$ system with coefficients in any holomorphic vector bundle from Example \ref{eg bg} (and the remarks after it) specialized to the manifold $X = \CC^d$.
One immediately checks that this is a holomorphically translation invariant free theory.
%where the operators $\eta_i = \frac{\partial}{\partial (\d z_i)}$ act in the natural way.
\end{eg}

\subsubsection{Translation invariant interactions}

Let's fix a general free holomorphically translation invariant theory $(V, (-,-)_V, Q^{hol})$ as above.
We now define what a holomorphically translation invariant interacting theory is.
Recall, translations span a $2d$-dimensional abelian Lie algebra $\CC^{2d} = \CC\left\{\frac{\partial}{\partial z_i}, \frac{\partial}{\partial \zbar_i}\right\}$. 
The first condition that an interaction be holomorphically translation invariant is that it be translation invariant, so invariant for this Lie algebra.
The additional condition is a bit more involved.

Let $\Bar{\eta}_i$ denote the operator on Dolbeault forms given by contraction with the antiholomorphic vector field $\frac{\partial}{\partial \zbar_i}$. 
Note that $\Bar{\eta}_i$ acts on the Dolbeault complex on $\CC^d$ with values in any vector bundle.
In particular it acts on the fields of a free holomorphically translation invariant theory as above, in addition to functionals on fields.

\begin{dfn}
A {\em holomorphically translation invariant} local functional is a translation invariant local functional $I \in \oloc(\sE_V)^{\CC^{2d}}$ such that $\Bar{\eta}_i I = 0$ for all $1 \leq i \leq d$. 
\end{dfn}

There is a succinct way of expressing holomorphic translation invariance as the Lie algebra invariants of a certain {\em dg Lie algebra}.
Denote by $\CC^d[1]$ the abelian $d$-dimensional graded Lie algebra in concentrated in degree $-1$ by the elements $\{\Bar{\eta}_i\}$.
We want to consider deformations that are invariant for the action by the total {\em dg} Lie algebra $\CC^{2d|d} = \CC^{2d} \oplus \CC^d[1]$.
The differential sends $\Bar{\eta}_i \mapsto \frac{\partial}{\partial \zbar_i}$.
The space of holomorphically translation invariant local functionals are denoted by $\oloc(\sE_V)^{\CC^{2d|d}}$.
The enveloping algebra of $\CC^{2d|d}$ is of the form
\ben
U(\CC^{2d|d}) = \CC \left[\frac{\partial}{\partial z_i},  \frac{\partial}{\partial \zbar_i}, \Bar{\eta}_i \right]
\een
with differential induced from that in $\CC^{2d|d}$. 
Note that this algebra is quasi-isomorphic to the algebra of constant coefficient polynomial holomorphic differential operators $\CC[\partial / \partial z_i] \xto{\simeq} U(\CC^{2d|d})$. 

A way of recasting the condition that a local function $I$ be both translation invariant and $\Bar{\eta}_i I = 0$ is to require it lie in the subspace of invariants for the dg Lie algebra $\CC^{2d|d}$. 
In other words, the space of holomorphically translation invariant local functionals is $\oloc(\sE_V)^{\CC^{2d|d}}$. 

From the definitions, we see that any translation invariant local functional is a sum of functionals of the form
\be\label{phifnl}
\varphi \mapsto \int_{\CC^d} F(D_1\varphi, \ldots, D_k \varphi) \d^d z
\ee
where $D_\alpha$ is an operator in the space 
\ben
\CC \left[\d \zbar_i, \frac{\partial}{\partial z_i}, \frac{\partial}{\partial \zbar_i}, \Bar{\eta}_i \right],
\een
and $F : \Omega^{0,*}(\CC^d, V)^{\tensor k} \to \Omega^{0,*}(\CC^d)$ is a linear map of the form
\ben
\Omega^{0,*}(\CC^d, V)^{\tensor k} \cong (\Omega^{0,*}(\CC^d) \tensor V_0)^{\tensor k} = \Omega^{0,*}(\CC^d)^{\tensor k} \tensor V_0^{\tensor k} \xto{\wedge \tensor F_0} \Omega^{0,*}(\CC^d),
\een
where $F_0 : V_0^{\tensor k} \to \CC$ is a linear map and $\wedge$ denotes the wedge product of forms.

The condition $\Bar{\eta}_i I = 0$ means that none of the $D_i$'s have any $\d \zbar_j$-dependence. 
Using this description we can exhibit the space of holomorphically translation functionals in a more efficient way. 
To state the result, we introduce a new class of local functionals (\ref{phifnl}) which only depend on differential operators $D_i$ built from $\frac{\partial}{\partial z}$, which we denote by $\sO^{hol, trans}_{\rm loc}(\sE_V)$. 
Like $\oloc(\sE_V)^{\CC^{2d|d}}$, such operators form a subspace
\[
\sO^{hol, trans}_{\rm loc}(\sE_V) \subset \oloc(\sE_V) .
\]
In fact, there is a natural inclusion $\sO^{hol, trans}_{\rm loc}(\sE_V) \subset  \oloc(\sE_V)^{\CC^{2d|d}}$ which arises from the fact that in the second space we allow for local functionals built from differential operators in the collection $\left\{\frac{\partial}{\partial z_i}, \frac{\partial}{\partial \zbar_i}, \Bar{\eta}_i \right\}$, whereas in the first space we only allow those built from $\frac{\partial}{\partial z_i}$. 

\begin{lem}\label{lem: hol trans local}
Let $(V, (-,-)_V, Q^{hol})$ be a free holomorphically translation invariant theory on $\CC^d$ and denote $\sE_V = \Omega^{0,*}(X, V)$. 
Then, the natural inclusion
\[
\left(\sO^{hol, trans}_{\rm loc}(\sE_V), \dbar + Q^{hol}\right) \hookrightarrow \left(\oloc(\sE_V)^{\CC^{2d|d}}, \dbar + Q^{hol}\right)
\]
is a quasi-isomorphism. 
%\ben
%H^*(\oloc(\sE_V)^{\CC^{2d|d}}, \dbar + Q^{hol}) \cong \left(\CC \cdot \d^d z \tensor^{\mathbb{L}}_{U(\CC^{2d|d})} \sO_{red} (J_0 E_V), Q^{hol} \right)
%\een
%where $E_V$ is the vector bundle on $\CC^d$ such that $\sE_V = \Gamma(E_V)$.
\end{lem}

\begin{proof}
We consider a spectral sequence in which we first take the cohomology with respect to $\dbar$. 
At the $E_1$-page, the isomorphism in cohomology follows from the quasi-isomorphism of dg Lie algebras $\CC^d \hookrightarrow (\CC^{2d|d}, \dbar)$. 
\end{proof}

This description of holomorphically translation invariant local functionals allows us to give a convenient description of deformations of holomorphically translation invariant theories. 
Suppose $(V,Q^{hol},(-,-)_V, I)$ be the data of an interacting holomorphically translation invariant theory on $\CC^d$.
We have already encountered the space of local functionals $\oloc(\sE_V)$ and the deformation complex of the interacting BV theory is
\ben
\Def_{\sE_V} = \left(\oloc(\sE_V), \dbar + Q^{hol} + \{I,-\}\right) .
\een
We'd like to characterize deformations that preserve holomorphically translation invariance. 

%Note that if $E$ is any vector bundle on $\CC^d$ we can consider the fiber at zero of its jet bundle that we denote $J_0 E$. 
Recall that in the holomorphic case there is the holomorphic jet bundle $J^{hol}V$.
The fiber at zero of this jet bundle may be identified as $J^{hol}_0 V = V_0 [[z_1,\ldots,z_d]]$ where the $z_i$'s denote the formal jet coordinate. 

\begin{cor}\label{cor: hol trans invt def}
Suppose that $Q^{hol} = 0$.
Then, there is a quasi-isomorphism
\ben
\left(\Def_{\sE_V}\right)^{\CC^{2d|d}} \simeq \CC \cdot \d^d z \tensor^{\LL}_{\CC[\partial_{z_1}, \ldots, \partial_{z_d}]} \sO_{red}(V_0[[z_1,\ldots,z_d]])[d].
\een
Equipped with differential $\{I^{hol},-\}$ where $I^{hol}$ only depends on holomorphic differential operators.
Here, $\partial_{z_i} = \frac{\partial}{\partial z_i}$ and $\CC \cdot \d^d z$ denotes the trivial right $\CC[\partial_{z_i}]$-module. 
\end{cor}

The local functional $I$ defining the classical holomorphic theory endows $J^{hol}V[-1]$ the structure of a $L_\infty$ algebra in $D_{\CC^d}$-modules. 
Repackaging the statement using Lie algebraic data we can rewrite the equivalence in the lemma as
\ben
\left(\Def_{\sE_V}\right)^{\CC^{2d|d}}\simeq \CC \cdot \d^d z \tensor^{\LL}_{\CC[\partial_{z_1}, \ldots, \partial_{z_d}]} \cred^*\left(V_0[[z]][-1])\right) [d].
\een

\begin{proof}

By Lemma \ref{lem: hol trans local} we have an expression for the holomorphically translation local functionals
\ben
\left(\Def_{\sE_V}\right)^{\CC^{2d|d}} = \left(\CC \cdot \d^d z \tensor_{U(\CC^{2d|d})} \sO_{red} (J_0 E_V)[d] , \dbar + \{I,-\}\right) .
\een
Since $\sO_{red}(J_0 E_V)$ is flat as a $U(\CC^{2d|d})$-module, it follows that we can replace the tensor product by the derived tensor product $\tensor^{\LL}$ up to quasi-isomorphism so that
\ben
\left(\Def_{\sE_V}\right)^{\CC^{2d|d}} \simeq \left(\CC \cdot \d^d z \tensor^{\LL}_{U(\CC^{2d|d})} \sO_{red} (J_0 E_V) [d] , \dbar + \{I,-\}\right) .
\een
Consider the complex $\left(\sO_{red}(J_0 E_V) , \dbar + \{I,-\}\right)$.
This complex is graded by symmetric degree, and the associated spectral sequence has first page the associated graded of $\sO_{red}(J_0 E_V)$ equipped with the $\dbar$ differential.
Moreover, at the $E_1$-page, we have the quasi-isomorphism
\ben
\left(\sO(J_0 E_V), \dbar\right)= \left(\sO_{red}(V_0 [[z_i, \zbar_i]][\d \zbar_i]), \dbar\right) \simeq \sO_{red}(V_0[[z_i]]) .
\een

Finally, we have already remarked that there is a quasi-isomorphism of algebras $U(\CC^{2d|d}) \simeq U(\CC^d)$ where the right-hand site is generated by the constant holomorphic vector fields. 
The proof of the claim follows. 

\end{proof}

%We end with an example.
%
%\begin{eg}
%The BV theory $\sE_V$ of holomorphic BF theory associated to any Lie algebra $\fg$ was discussed in Example \brian{ref}.
%This theory exists on any complex manifold, but when placed on $\CC^d$ it is naturally holomorphically translation invariant.
%The classical theory induces the structure of a local $L_\infty$ algebra on the shifted complex of fields$\sE_V[-1]$, and the deformation complex can be written as the local Lie algebra cohomology $\cloc^*(\sE_V[-1])$.
%
%The dg Lie algebra $\sL = \Omega^{0,*}(X, \fg)$ has the natural structure of a local Lie algebra on any complex manifold $X$.
%Its $!$-dual $\sL^! = \Omega^{d,*}(X, \fg^*)[d]$ has the natural structure of a $\sL$-module where $\fg$ acts on $\fg^*$ by the coadjoint and the Dolbeualt forms act on themselves via wedge product.
%Note that as a local $L_\infty$ algebra $\sE_V[-1]$ can be written as an extension
%\ben
%0 \to \sL^![-3] \to \sE_V[-1] \to \sL \to 0 .
%\een
%\end{eg}

The holomorphic BF system, as in Example \ref{eg: bf}, on $X = \CC^d$ is an example of a holomorphically translation invariant theory. 
So is the topological BF system, as in Example \ref{eg: bftop}. 

\begin{eg} {\em Holomorphic superpotential.} 
This is a different flavor of a holomorphically translation invariant theory involving the $\beta\gamma$ system and is largely motivated by physics. 
Consider the $\beta\gamma$ system on $\CC^d$ with values in $V$. 
In addition, let $W \in \CC[V]$ be a polynomial on the vector space $V$. 
Then, $W$ extends in a natural way to a Dolbeault valued functional on $\Omega^{0,*}(\CC^d) \tensor V$. 
One defines the local functional
\[
I_W (\beta, \gamma) = \int_{\CC^d} \d^d z \; W(\gamma) .
\] 
It is immediate to see that $I_W$ is holomorpically translation invariant. 
On the other hand, it is not, in general, a degree {\em zero} functional. 
Hence, it does not define a classical theory in the usual sense. 
It does, however, define a slightly weaker classical theory that is only $\ZZ/2$ graded rather than the usual $\ZZ$ grading we are accustomed to.
 
We don't develop the formal definition here, but $I_W$ defines a holomorpically translation invariant $\ZZ/2$-graded BV theory. 
When $d=2$, the $\beta\gamma$ system arises as the minimal twist of the free $\cN = 1$ chiral supermultiplet on $\RR^4$. 
In the presence of the interaction $I_W$, the theory is equal to the minimal twist of the $\cN=1$ chiral multiplet with {\em holomorphic superpotential} given by $W$. \footnote{In super language, the superpotential term is usually written as $\int d^2 \theta \int \d^4 x W(\Phi)$, where $\Phi$ is the chiral superfield.}
\end{eg}

\section{One-loop regularization for theories on $\CC^d$} \label{sec: hol renorm}

In Wilsonian's approach to quantum field theory, constructing the path integral involves exhibiting a family of theories parametrized by some scale $L > 0$, that we take for illustration to be in units of length. 
The main idea is that the theory at scale $L$ describes all interactions happening at length scales smaller than or equal to $L$. 
To obtain the full perturbative QFT, one takes the limit $L \to \infty$, where all quantum interactions are included. 
In practice, one has a good handle on the theory between some finite scales $\epsilon < L$, and to obtain the theory at scale $L$ one must make sense of the $\epsilon \to 0$ limit. 
Generally speaking, the naive limit is ill-defined; this is the part of the strategy for constructing a QFT where renormalization comes in. 

In this section we consider the renormalization of holomorphic field theories on $\CC^d$, for general $d \geq 1$. 
We start with a classical holomorphic theory on $\CC^d$ and study its one-loop homotopy renormalization group flow from some finite scale $\epsilon$ to scale $L$.
This is where the theory is completely well-defined. 
Explicitly, this flow manifests as a sum over weights of graphs; that is, {\em Feynman diagrams}.
In terms of diagrams, we consider the sum over graphs of genus at most one where at each vertex we place the holomorphic interaction defining the classical theory.
The edges of the graphs are labeled by the propagator, which, for us, is an effective replacement for the Green's function of the $\dbar$ operator defining the kinetic piece of the holomorphic field theory.

To obtain a quantization of a classical theory one must make sense of the $\epsilon \to 0$ limit of this construction. 
In general, this involves introducing a family of {\em counterterms}.
The presence of counterterms can be an often undesirable, but necessary part of constructing a quantum field theory.
On one hand, logarithmic counterterms encode the $\beta$-function of an interacting field theory, which is a sensitive invariant and is important quantity to experimentally measure quantities in QFT.
Roughly, this quantity measures how couplings run with renormalization group flow. 
Counterterms can also be extremely unwieldy. 
For instance, some theories of gravity require the introduction of infinitely many such counterterms \cite{THooftVelt}.
In this paper, we show how holomorphic theories on flat space are as well-behaved as possible when it comes to renormalization. 

Our main result in this section is the following (which we state more carefully in Theorem \ref{thm: holrenorm3} below):

\begin{thm} \label{thm: holrenorm2}
For a holomorphic theory on $\CC^d$, there exists a one-loop (pre)quantization where the naive $\epsilon \to 0$ limit exists and no counterterms are required.
\end{thm}

\begin{rmk}
Already, in \cite{LiVertex} Li has proved a stronger version of Theorem \ref{thm: holrenorm2} when the complex dimension is $d =1$.
His result holds to all orders in $\hbar$, and applies it to give an elegant interpretation of the quantum master equation for chiral conformal field theories on (flat) Riemann surfaces using vertex algebras.
Although we do not make any statements in this thesis past one-loop
quantizations, the higher loop behavior remains a rich and subtle problem that we hope to return to.
\end{rmk}

As a peculiar corollary of our main result, and our work in developing the one-loop $\beta$-function for QFT in the BV formalism \cite{EWY}, we have the following. 

\begin{cor}
The one-loop $\beta$-function of a holomorphic theory on $\CC^d$ is identically zero. 
\end{cor}

This corollary has ``no-go" style consequences for twists of supersymmetric field theories. 
As we have already mentioned, often times a supersymmetric field theory on $\RR^{2d}$ admits a holomorphic twist where half of the translations are left $Q$-exact.
This result implies that the $\beta$-function is not protected under such holomorphic twists. 
For instance, $\cN=1$ supersymmetric Yang-Mills on $\RR^4$ admits a holomorphic twist to holomorphic BF theory. 
While Yang-Mills has a non-trivial $\beta$-function, our results show that the $\beta$-function for holomorphic BF theory is zero. 

The proof of the main result will be involve explicit evaluations and estimates of weights of Feynman diagrams. 
Before proceeding with the core analysis, we set up the problem using our notation and conventions used above. 

Suppose $(V, Q^{hol}, (-,-)_V)$ prescribes the data of a free holomorphic theory on $\CC^d$.
This means that $V$ is a holomorphic bundle on $\CC^d$, $Q^{hol} : \sV^{hol} \to \sV^{hol}$ is a holomorphic differential operator, and $(-,-)_V$ is a (shifted) $K_{\CC^d}$-valued pairing on $V$. 
We assume, in addition, that $Q^{hol}$ is translation invariant.
Concretely, this means that
\ben
Q^{hol} \in \CC \left[\frac{\partial}{\partial z_1} , \ldots, \frac{\partial}{\partial z_d}\right].
\een

The complex of fields, in the BV formalism, are given by the following deformed Dolbeault complex
\ben
\sE_V = \left(\Omega^{0,*}(\CC^d, V), \dbar + Q^{hol}\right) .
\een
We will fix a trivialization for the holomorphic vector bundle $V = \CC^d \times V_0$, where $V_0$ is the fiber over $0 \in \CC^d$.
This leads to an identification $\Omega^{0,*}(\CC^d , V) = \Omega^{0,*}(\CC^d) \tensor_\CC V_0$.
Further, we write the $(-1)$-shifted symplectic structure defining the classical BV theory in the form
\ben
\omega_V(\alpha \tensor v, \beta \tensor w) = (v,w)_{V_0} \int \d^d z (\alpha \wedge \beta)
\een
where $(-,-)_{V_0}$ is a degree $(d-1)$-shifted pairing on the finite dimensional vector space $V_0$. 

A holomorphic interacting theory is prescribed by a holomorphic Lagrangian $I^{hol} \in \sO_{\rm loc}^{hol,+}(V)$, see Definition \ref{dfn: plus}. 
As we have seen in Section \ref{sec: interacting} any holomorphic Lagrangian determines a local functional on its Dolbeualt complex via integration $I = \int_X I^{\Omega^{0,*}}$. 
Here, as above, the notation $I^{\Omega^{0,*}}$ denotes the canonical extension of $I^{hol}$ to the Dolbeualt complex for $V$. 
Using the trivialization $V = \CC^d \times V_0$ and for the translation invariant holomorphic top forms by $\CC \cdot \d^d z$, we can express the local functional as
\ben
I_k (\alpha) = \int I^{hol}_k(\alpha) = \int D_{k,1}(\phi_{k,1} (\alpha))\cdots D_{k,k} ( \phi_{k,k}(\alpha)) \d^d z 
\een
where each $D_{i,j}$ is a holomorphic differential operator $D_{i,j} \in \CC\left[\frac{\partial}{\partial z_i}\right]$, and $\phi_{i,j} \in V_0^\vee$.

\subsection{Homotopy RG flow}

As we've already mentioned, the main goal of this section is to show that for holomorphic theories on $\CC^d$ the one-loop renormalization group flow produces a prequantization modulo $\hbar^2$. 
We follow the terminology of \cite{CosRenorm} and use {\em prequantization} to refer to an effective family of functionals satisfying renormalization group flow but not necessarily the quantum master equation. 
We will see consequences of our result for solving the quantum master equation modulo $\hbar^2$ in the next section.  

The building block in Costello's approach to renormalization is an effective family of functionals $\{I[L]\}$ parametrized by a {\em length scale} $L > 0$. 
For each $L > 0$ the functional $I[L] \in \sO(\sE)[[\hbar]]$ must satisfy various conditions, which are carefully stated in Definition 8.2.9.1 of \cite{CG2}. 
We will recall some key aspects that will be useful for our purposes. 
The main condition is a compatibility between the functionals $I[L]$ as one changes the length scale; this is referred to as {\em homotopy renormalization group (RG) flow}.
The flow from scale $L>0$ to $L'>0$ is encoded by an invertible linear map
\ben
W(P_{L < L'} , -) : \sO^+(\sE)[[\hbar]] \to \sO^+(\sE)[[\hbar]]
\een
defined as a sum over weights of graphs $W (P_{L<L'}, I) = \sum_{\Gamma} W_{\Gamma}(P_{L<L'}, I)$. 
Here, $\Gamma$ denotes a graph, and the weight $W_\Gamma$ is defined as follows.
One labels the vertices of valence $k$ by the $k$th homogenous component of the functional $I$. 
The edges of the graph are labeled by the propagator $P_{L<L'}$.
The total weight is given by iterative contractions of the homogenous components of the interaction with the propagator. 
For a more precise definition see Chapter 2 of \cite{CosRenorm}.

The family of functionals $\{I[L]\}$ defining a quantization must satisfy the {\em RG flow equation}
\ben
I[L'] = W(P_{L<L'}, I[L])
\een
for all $L < L'$. 
Given a classical interaction $I \in \oloc(\sE)$, there is a natural way to attempt construct an effective family of functionals satisfying the RG flow equations.
Indeed, it follows from elementary properties of the homotopy RG flow operator $W(P_{L < L'}, -)$ that {\em if} the functional
\ben
I[L] \;\; ``=" \;\; W(P_{0<L}, I) 
\een
were to be well-defined for each $L >0$, then the RG flow equations would automatically be satisfied for the collection $\{I[L]\}$. 
The problem is that this naive guess is ill-defined due to the distributional nature of the propagator $P_{0<L}$. 
The approach of Costello is to introduce a small parameter $\epsilon > 0$ and to consider the limit of the functionals $W(P_{\epsilon < L}, I)$ as $\epsilon \to 0$. 
For most theories, this $\epsilon \to 0$ limit is ill-defined, but there always exist $\epsilon$-dependent {\em counterterms} $I^{CT}(\epsilon)$ rendering the existence of the $\epsilon \to 0$ limit of $W(P_{\epsilon < L}, I - I^{CT}(\epsilon))$. 

Our main goal in this section amounts to showing that the naive $\epsilon \to 0$ limit exists without the necessity to introduce counterterms. 
This is a salient feature of holomorphic theories on $\CC^d$ that we will take advantage of to characterize anomalies, for instance. 

We will only consider quantizations defined modulo $\hbar^2$.
In this case, the homotopy RG flow takes the explicit form:
\ben
W(P_{\epsilon<L}^V , I) = \sum_{\Gamma} \frac{\hbar^{g(\Gamma)}}{|{\rm Aut}(\Gamma)|} W_\Gamma (P_{\epsilon<L}^V, I) .
\een
The sum is over graphs of genus $\leq 1$ and $W_\Gamma$ is the weight associated to the graph $\Gamma$. 

We can now state the main result of this section.

\begin{thm}\label{thm: holrenorm3}
Let $\sE$ be a holomorphic theory on $\CC^d$ with classical interaction $I^{cl}$.  
Then, there exists a one-loop prequantization $\{I[L] \; | \; L > 0\}$ of $I^{cl}$ involving no counterterms. 
That is, the $\epsilon \to 0$ limit of
\ben
W(P_{\epsilon<L} , I) \mod \hbar^2 \in \sO(\sE) [[\hbar]] / \hbar^2
\een
exisits.
Moreover, if $I$ is holomorphically translation invariant we can pick the family $\{I[L]\}$ to be holomorphically translation invariant as well.
\end{thm}

\subsection{Holomorphic gauge fixing}
 
The next component of a prequantization is the choice of a gauge fixing condition.
From a physics point of view the choice of a gauge fixing condition is common place when computing quantities 
in QFT. 
Mathematically, it is equivalent to choosing an isotropic subspace of the space of fields which is necessary to define the path integral in the BV formalism.
In our philosophy of QFT, all theories are really defined over the space (or simplicial set) of gauge fixing conditions. 
The theory does not depend on a gauge fixing condition in the sense that a path in the space of gauge fixing conditions leads to a homotopy between the associated theories. 
See Chapter 5 of \cite{CosRenorm} for a thorough formulation of this. 

In our approach, a gauge fixing condition appears through the choice  gauge fixing operator is a square-zero operator on fields
\ben
Q^{GF} : \sE_V \to \sE_V[-1],
\een
of cohomological degree $-1$ such that $[Q, Q^{GF}]$ is a generalized Laplacian on $\sE$ where $Q$ is the linearized BRST operator. 
For a complete definition see Section 8.2.1 of \cite{CG2}.

For holomorphic theories there is a convenient choice for a gauge fixing operator. 
To construct it we fix the standard flat metric on $\CC^d$. 
Doing this, we let $\dbar^*$ be the adjoint of the operator $\dbar$.
Using the coordinates on $(z_1,\ldots, z_d) \in \CC^d$ we can write this operator as
\ben
\dbar^* = \sum_{i=1}^d \frac{\partial}{\partial (\d \zbar_i)} \frac{\partial}{\partial z_i} .
\een
The operator $\frac{\partial}{\partial (\d \zbar_i)}$ is the contraction with the anti-holomorphic vector field $\frac{\partial}{\partial \zbar_i}$. 
The operator $\dbar^*$ extends to the complex of fields via the formula
\ben
Q^{GF} = \dbar^* \tensor {\rm id}_V : \Omega^{0,*}(X , V) \to \Omega^{0,*-1}(X, V),
\een

\begin{lem}
The operator $Q^{GF} = \dbar^* \tensor {\rm id}_V$ is a gauge fixing operator for the free theory $(\sE_V, \dbar + Q, \omega_V)$.
\end{lem}
\begin{proof}
Clearly, $Q^{GF}$ is square zero since $(\dbar^*)^2 = 0$.
Since $Q^{hol}$ is a translation invariant holomorphic differential operator we have
\ben
[\dbar + Q^{hol}, Q^{GF}] = [\dbar,\dbar^*] \tensor \id_{V} .
\een
The operator $[\dbar,\dbar^*]$ is the Dolbeault Laplacian $\Delta_{\dbar}$ on $\CC^d$, which 
in coordinates is
\ben
\Delta_{\dbar} = -\sum_{i=1}^d \frac{\partial}{\partial \zbar_i}\frac{\partial}{\partial z_i} .
\een
In particular, the operator $[\dbar,\dbar^*] \tensor \id_{V}$ is a generalized Laplacian. 

Finally, we must show that $Q^{GF}$ is (graded) self-adjoint for the shifted symplectic pairing $\omega_V$. 
This follows from the fact about Dolbeualt forms on $\CC^d$.
If $\alpha,\beta \in \Omega^{0,*}_c(\CC^d)$ then
\ben 
\int_{\CC^d} (\dbar^* \alpha) \wedge \beta \wedge \d^d z = \pm \int_{\CC^d} \alpha \wedge (\dbar^* \beta) \wedge \d^d z .
\een
\end{proof}

\begin{rmk}
One may ask what happens if we choose a different metric on $\CC^d$ to define the gauge fixing operator. 
For every choice of a Hermitian metric $h$ on $\CC^d$ we obtain an operator $\dbar^*_h$ and hence a gauge fixing condition. 
In fact, this defines a {\em family} of theories defined over the space of all Hermitian metrics. 
Since this space is affine, hence connected, we can always choose a path to the standard metric to any other one, thus resulting in a homotopy equivalence between quantizations defined by the standard metric and the fixed one. 
The subtlety here is that the quantization provided by an arbitrary Hermitian metric may not be as simple as the one for the flat metric. 
In fact, there may be one-loop divergences. 
Nevertheless, the homotopical framework for QFT developed in \cite{CosRenorm} implies that the quantizations associated to two different Hermitian metrics will be equivalent. 
\end{rmk}

\subsection{The propagator on $\CC^d$}

The gauge fixing operator determines a generalized Laplacian, which for us is essentially the ordinary Dolbeault Laplacian on $\CC^d$. 
Our regularization scheme utilizes the heat kernel associated to the Laplacian, for which we recall the explicit form below.
By definition, the scale $L>0$ heat kernel is a symmetric element $K_L^V \in \sE_V(\CC^d) \tensor \sE_V(\CC^d)$ that satisfies
\ben
\omega_V(K_L, \varphi) = e^{-L[Q,Q^{GF}] } \varphi
\een
for any field $\varphi \in \sE_V$.
Thus, it is an integral kernel for the operator $e^{-L[Q,Q^{GF}]}$. 
For a more detailed definition of how heat kernels are used to defined a quantum field theory in the BV formalism, see Section 8.2.3 in \cite{CG2}. 
In this section we deduce the explicit form of the heat kernel for our holomorphic theory on $\CC^d$. 

The tensor square of $\sE_V(\CC^d)$ decomposes as 
\be\label{splitting}
\sE_V(\CC^d) \tensor \sE_V(\CC^d) = \left(\Omega^{0,*}(\CC^d) \tensor \Omega^{0,*}(\CC^d)\right) \tensor (V_0 \tensor V_0) .
\ee
We will decompose the heat kernel accordingly. 

Pick a basis $\{e_i\}$ of $V_0$ and let 
\ben
{\bf C}_{V_0} = \sum_{i,j} \omega_{ij} (e_i \tensor e_j) \in V_0 \tensor V_0
\een
be the quadratic Casimir.
Here, $(\omega_{ij})$ is the inverse matrix to the pairing $(-,-)_{V_0}$. 

Due to the nature of our symplectic pairing, we see that the heat kernel splits with respect to the decomposition in Equation (\ref{splitting}) as
\ben
K_{L}^V (z,w) = K^{an}_L(z,w) \cdot {\bf C}_{V_0} .
\een
The analytic part $K^{an}_L$ is independent of $V$ and equal to the heat kernel for Dolbeault Laplacian $\Delta_{\dbar}$ acting on Dolbeault forms on $\CC^d$. 

We can further split this analytic heat kernel as the heat kernel for the ordinary Laplacian acting on functions. 
Indeed, for $L>0$ the analytic heat kernel $K_L^{an}$ is equal to
\ben
K_L^{an} (z,w) = k_L^{an}(z,w) \prod_{i=1}^d (\d \zbar_i - \d \wbar_i)  \in \Omega^{0,*} (\CC^d) \tensor \Omega^{0,*} (\CC^d) \cong \Omega^{0,*} (\CC^d \times \CC^d) \cong C^\infty(\CC^d \times \CC^d)[\d z, \d w]
\een
where $k_L^{an}(z,w) \in C^\infty(\CC^d \times \CC^d)$ is the heat kernel for the Laplacian acting on functions. 
It is normalized by the rule
\[
(e^{- L \Delta_{\dbar}} f)(z) = \int_{w \in \CC^d} \d^{2d} w \; k_L^{an}(z,w) f(w)
\]
where $f \in C^\infty(\CC^d)$. 
Explicitly, $k_L^{an}$ is given by
\[
k^{an}_L(z,w) = \frac{1}{(2\pi i L)^d} e^{-|z-w|^2/4L}  .
\]

The propagator for the holomorphic theory $\sE_V$ is defined using the heat kernels above by the equation
\ben
P_{\epsilon < L}^V(z,w) = \int_{t=\epsilon}^L \d t (Q^{GF} \tensor 1) K_{L}^V(z,w) .
\een
Since the element ${\bf C}_{V_0}$ is independent of the coordinate on $\CC^d$, the propagator also decomposes as 
\[
P_{\epsilon < L}^V(z,w) = P_{\epsilon < L}^{an}(z,w) \cdot {\bf C}_{V_0}
\]
where
\begin{align*}
P_{\epsilon < L}^{an}(z,w) & = \int_{t=\epsilon}^L \d t (\dbar^* \tensor 1) K_{L}^V(z,w) \\
& = \int_{t=\epsilon}^L \d t \frac{1}{(2 \pi i t)^d} \sum_{j=1}^d (-1)^{j-1}  \left(\frac{\zbar_j - \wbar_j}{4 t} \right)  e^{-|z-w|^2 / 4t}  \prod_{i \ne j}^d (\d \zbar_i - \d \wbar_i) .
\end{align*}

The propagator $P_{\epsilon <L}$ is an effective replacement for the Green's function for $\dbar$ on $\CC^d$. 
In the limit as $\epsilon \to 0$ and $L \to \infty$, this propagator reduces to the Green's function for the Dobleault operator on $\CC^d$.
We can see this simplification explicitly. 
%We see that the differential form part above is proportional to the Bochner-Martinelli kernel $\omega_{BM} \in \Omega^*(\CC^d \times \CC^d \setminus \Delta)$ 
%\ben
%\omega_{BM}(z,w) = C_d \frac{1}{|z-w|^{2d}} \sum_{j} (-1)^{j-1} (\zbar_j - \wbar_j) (\d^d z - \d^d w) \prod_{i \ne j} (\d \zbar_i - \d \wbar_i) .
%\een
%where $C_d = (d-1)! / (2 \pi i)^d$ is a constant depending only on the dimension $d$.
%A simple corollary of the above calculation is the following fact that we will use later on in Section \ref{sec: local obs}.

First, we recall the form of the Green's function.
Introduce the $\delta$-distribution $\delta_\Delta$ along the diagonal in $\CC^d \times \CC^d$.
In formulas
\[
\delta_{\Delta} : \Omega^{0,*}_c(\CC^d) \times \Omega^{0,*}_c \to \CC \;\; , \;\; (\alpha, \beta) \mapsto \int_{\Delta \subset \CC^d \times \CC^d} \d^d z \wedge \alpha \wedge \beta.
\]
The Green's function for $\dbar$ is given in terms of the {\em Bochner-Martinelli kernel}.
To define it, first consider the smooth form on $\CC^d \times \CC^d$ away from the diagonal
$\omega_{BM} \in \Bar{\Omega}^{0,*}(\CC_z^d \times \CC_w^d \setminus \Delta)$ given by
\[
\omega_{BM}(z,w) = \frac{(d-1)!}{(2\pi i)^d} \frac{1}{|z-w|^{2d}} \sum_{i=1}^d (-1)^{i-1} (\zbar_i - \wbar_i) \prod_{j \ne i} (\d \zbar_j - \d \wbar_j) .
\]
Since $\dbar \omega_{BM}(z,w) = 0$ away from the diagonal, we see that $\dbar \omega_{BM}$ extends to a distribution form on $\CC^d \times \CC^d$. 
Indeed, this distribution solves Green's equation for the $\dbar$-operator
\[
\dbar \omega_{BM} \wedge \d^d z = \delta_{\Delta} .
\] 

For more details on the above kernel and its relation to higher residues, see Chapter 3 of \cite{GriffithsHarris} for instance.
For now, we have the immediate calculation. 

\begin{lem} \label{lem: bm}
The $\epsilon \to 0, L\to \infty$ distributional limit of the propagator $P_{\epsilon<L}(z,w)$ exists.
Moreover, as distributions
\ben
\lim_{\epsilon \to 0} \lim_{L \to \infty} P_{\epsilon<L}(z,w) = \omega_{BM}(z,w) .
\een
\end{lem}
\begin{proof}
Note that
\bestar
P_{\epsilon < L}(z,w) & = & \int_{t = \epsilon}^L \d t e^{- |z-w|^2 / 4t} \frac{1}{(2 \pi i t)^d} \sum_{j = 1}^d (-1)^{j-1} \frac{\zbar_j - \wbar_j}{4t} \prod_{i \ne j} (\d \zbar_i - \d \wbar_i) \\ & = & \frac{1}{(2 \pi i)^d} \frac{1}{|z-w|^{2d}} \sum_{j} (-1)^{j-1} (\zbar_j - \wbar_j) \\ & & \times \prod_{i \ne j} (\d \zbar_i - \d \wbar_i) \int_{u = |z-w|^2/L}^{|z-w|^2/\epsilon} \d u u^{d-1} e^{-u} .
\eestar
In the second line we have made the substitution $u = |z-w|^2 / 4t$.
Integration over $u$ produces the desired result. 
\end{proof}

\subsection{Trees}

We now turn to studying the one-loop effective action for the holomorphic theory on $\CC^d$. 
For the genus zero graphs, or trees, we do not have any analytic difficulties to worry about. 
The propagator $P_{\epsilon<L}^V$ is smooth so long as $\epsilon,L > 0$ but when $\epsilon \to 0$ it inherits a singularity along the diagonal $z = w$.
This is what contributes to the divergences in the naive definition of RG flow $W(P_{0<L}, -)$.
But, if $\Gamma$ is a tree the weight $W_\Gamma(P_{0<L}^V, I)$ only involves multiplication of distributions with transverse singular support, so is well-defined.
Thus we have observed the following.

\begin{lem} 
If $\Gamma$ is a tree then $\lim_{\epsilon \to 0} W_{\Gamma}(P_{\epsilon < L}, I)$ exists.
\end{lem}

The only possible divergences in the $\epsilon \to 0$ limit, then, must come from graphs of genus one, which we now direct our attention to.

\subsection{A simplification for one-loop weights}

Every graph of genus one is a wheel with some trees protruding from the external edges of the tree.
Thus, we can write the weight of a genus one graph as a product of weights associated to trees times the weight associated to a wheel.
We have just observed that the weights associated to trees are automatically convergent in the $\epsilon \to 0$ limit, thus it suffices to focus on genus one graphs that are purely wheels with some number of external edges.

The definition of the weight of the wheel involves placing the propagator at each internal edge and the interaction $I$ at each vertex. 
The weights are evaluated by placing compactly supported fields $\varphi \in \sE_{V,c} = \Omega^{0,*}_c(\CC^d, V)$ at each of the external edges.
We will make two simplifications:
\begin{enumerate}
\item the only $\epsilon$ dependence appears in the analytic part of the propagator $P_{\epsilon<L}^{an}$, so we can forget about the combinatorial factor ${\bf C}_{V_0}$ and assume all external edges are labeled by compactly supported Dolbeault forms in $\Omega^{0,*}_c(\CC^d)$;
\item each vertex labeled by $I$ is a sum of interactions of the form
\ben
\int_{\CC^d} D_1(\varphi) \cdots D_k(\varphi) \d^d z
\een
where $D_i$ is a holomorphic differential operator (only involves $\frac{\partial}{\partial z_i}$-derivatives). 
Some of the differential operators will hit the compactly supported Dolbeault forms placed on the external edges of the graph.
The remaining operators will hit the internal edges labeled by the propagators.
Since a holomorphic differential operator preserves the space of compactly supported Dolbeault forms that is independent of $\epsilon$, we replace each input by an arbitrary compactly supported Dolbeault form.
\end{enumerate}

Thus, for the $\epsilon \to 0$ behavior it suffices to look at weights of wheels with arbitrary compactly supported functions as inputs where each of the internal edges are labeled by some translation invariant holomorphic differential operator 
\ben
D = \sum_{n_1,\ldots n_d} \frac{\partial^{n_1}}{\partial z_{1}^{n_1}}\cdots \frac{\partial^{n_d}}{\partial z_{d}^{n_d}}
\een
applied to the propagator $P_{\epsilon<L}^{an}$.
This motivates the following definition. 

\begin{dfn}\label{dfn: analytic weight}
Let $\epsilon , L > 0$. 
In addition, fix the following data.
\begin{enumerate}[(a)]
\item An integer $k \geq 1$ that will be the number of vertices of the graph.
\item For each $\alpha = 1, \ldots, k$ a sequence of integers
\ben
\vec{n}^\alpha = (n_1^\alpha, \ldots, n_d^{\alpha}) .
\een
We denote by $(\vec{n}) = (n_{i}^j)$ the corresponding $d \times k$ matrix of integers. 
\end{enumerate}
The analytic weight associated to the pair $(k, (\vec{n}))$ is the smooth distribution
\ben
W_{\epsilon < L}^{k, (n)} : C_c^\infty((\CC^d)^k) \to \CC,
\een
that sends a smooth compactly supported function $\Phi \in C_c^\infty((\CC^d)^k) = C_c^\infty(\CC^{dk})$ to
\be\label{weight1}
W_{\epsilon < L}^{k, (n)} (\Phi) = \int_{(z^1,\ldots, z^k) \in (\CC^d)^k} \prod_{\alpha=1}^k \d^d z^\alpha \Phi(z^1,\ldots,z^k) \prod_{\alpha = 1}^k \left(\frac{\partial}{\partial z^\alpha}\right)^{\vec{n}^\alpha} P_{\epsilon < L}^{an}(z^\alpha, z^{\alpha+1}) .
\ee
In the above expression, we use the convention that $z^{k+1} = z^1$. 
\end{dfn}

The coordinate on $(\CC^{d})^k$ is given by $\{z_i^\alpha\}$ where $\alpha = 1,\ldots,k$ and $i = 1, \ldots, d$. 
For each $\alpha$, $\{z_1^\alpha, \ldots, z_d^\alpha\}$ is the coordinate for the space $\CC^d$ sitting at the vertex labeled by $\alpha$. 
We have also used the shorthand notation
\ben
\left(\frac{\partial}{\partial z^\alpha}\right)^{\vec{n}^\alpha} = \frac{\partial^{n^\alpha_1}}{\partial z^\alpha_1} \cdots  \frac{\partial^{n^\alpha_d}}{\partial z^\alpha_d}.
\een

We will refer to the collection of data $(k, (\vec{n}))$ in the definition as {\em wheel data}.
The motivation for this is that the weight $W_{\epsilon < L}^{k, (n)}$ is the analytic part of the full weight $W_{\Gamma}(P^V_{\epsilon<L}, I)$ where $\Gamma$ is a wheel with $k$ vertices. 

We have reduced the proof of Proposition \ref{thm: holrenorm3} to showing that the $\epsilon \to 0$ limit of the analytic weight $W_{\epsilon < L}^{k, (\vec{n})}(\Phi)$ exists for any choice of wheel data $(k, (\vec{n}))$.
To do this, there are two steps. 
First, we show a vanishing result that says when $k \leq d$ the  weights vanish for purely algebraic reasons. 
The second part is the most technical aspect of the chapter where we show that for $k > d$ the weights have nice asymptotic behavior as a function of $\epsilon$.

\begin{lem} Let $(k, (\vec{n}))$ be a pair of wheel data.
If the number of vertices $k$ satisfies $k \leq d$ then
\ben
W_{\epsilon < L}^{k, (n)}  = 0
\een
as a distribution on $\CC^{dk}$ for any $\epsilon,L > 0$. 
\end{lem}

The proof of this lemma is essentially identical to the proof of the Claim on page 73 of \cite{bcov}, in the context of BCOV theory on odd dimensional Calabi-Yau manifolds. 
The upshot is that the method of proof works for general holomorphic theories on $\CC^d$ which we consider here. 

\begin{proof}
In the integral expression for the weight (\ref{weight1}) there is the following factor involving the product over the edges of the propagators:
\be\label{productprops2}
\prod_{\alpha = 1}^k \left(\frac{\partial}{\partial z^\alpha}\right)^{\vec{n}^\alpha} P_{\epsilon < L}^{an}(z^\alpha, z^{\alpha}) .
\ee
We will show that this expression is identically zero.
To simplify the expression we first make the following change of coordinates on $\CC^{dk}$:
\begin{align}
w^\alpha & = z^{\alpha+1} - z^\alpha \;\;\; , \;\;\; 1\leq \alpha < k \label{coords1}\\
w^k & = z^k \label{coords2} .
\end{align}
Introduce the following operators
\ben
\eta^\alpha = \sum_{i=1}^{d} \wbar_i^\alpha \frac{\partial}{\partial (\d \wbar_i^\alpha)}
\een
acting on differential forms on $\CC^{dk}$.
The operator $\eta^\alpha$ lowers the anti-holomorphic Dolbuealt type by one : $\eta : (p,q) \to (p,q-1)$.
Equivalently, $\eta^\alpha$ is contraction with the anti-holomorphic Euler vector field $\wbar_i^\alpha \partial / \partial \wbar_i^\alpha$.

%\brian{add something like “We will show that the integrand of the graph integral is not a top form and hence the integral is manifestly zero. To do this, we need to identify the Dolbeault type of the form ...”}

Once we do this, we see that the expression (\ref{productprops2}) can be written as 
\ben
\left(\left(\sum_{\alpha=1}^{k-1} \eta^\alpha \right) \prod_{i=1}^d \left(\sum_{\alpha = 1}^{k-1} \d \wbar_{i}^\alpha\right) \right) \prod_{\alpha=1}^{k-1}\left( \eta^\alpha \prod_{i=1}^d \d \wbar_i^\alpha\right) .
\een
Note that only the variables $\wbar_i^{\alpha}$ for $i=1,\ldots,d$ and $\alpha = 1,\ldots, k-1$ appear. 
Thus we can consider it as a form on $\CC^{d(k-1)}$.
As such a form it is of Dolbeault type $(0, (d-1) + (k-1)(d-1)) = (0, (d-1)k)$. 
If $k < d$ then clearly $(d-1)k > d(k-1)$ so the form has greater degree than the dimension of the manifold and hence it vanishes. 

The case left to consider is when $k = d$.
In this case, the expression in (\ref{productprops2}) can be written as
\be\label{productprops1}
\left(\left(\sum_{\alpha=1}^{d-1} \eta^\alpha \right) \prod_{i=1}^d \left(\sum_{\alpha = 1}^{d-1} \d \wbar_{i}^\alpha\right) \right) \prod_{\alpha=1}^{d-1}\left( \eta^\alpha \prod_{i=1}^d \d \wbar_i^\alpha\right) .
\ee
Again, since only the variables $\wbar_i^{\alpha}$ for $i=1,\ldots,d$ and $\alpha = 1,\ldots, d-1$ appear, we can view this as a differential form on $\CC^{d(d-1)}$. 
Furthermore, it is a form of type $(0, d(d-1))$. 
For any vector field $X$ on $\CC^{d(d-1)}$ the interior derivative $i_X$ is a graded derivation. 
Suppose $\omega_1,\omega_2$ are two $(0,*)$ forms on $\CC^{d(d-1)}$ such that the sum of their degrees is equal to $d^2$. 
Then, $\omega_1 \iota_X \omega_2$ is a top form for any vector field on $\CC^{d(d-1)}$.
Since $\omega_1 \omega_2 = 0$ for form type reasons, we conclude that $\omega_1 \iota_X \omega_2 = \pm (i_X \omega_1) \omega_2$ with sign depending on the dimension $d$. 
Applied to the vector field $\zbar_i^1\partial / \partial \wbar_i^1$ in (\ref{productprops1}) we see that the expression can be written (up to a sign) as 
\ben
\eta^1 \left(\sum_{\alpha=1}^{d-1} \eta^\alpha \prod_{i=1}^d \left(\sum_{\alpha = 1}^{d-1} \d \wbar_{i}^\alpha\right) \right) \left(\prod_{i=1}^d \d \wbar_i^1\right) \prod_{\alpha=2}^{d-1} \left( \eta^\alpha \prod_{i=1}^d \d \wbar_i^\alpha\right) .
\een
Repeating this, for $\alpha =2,\ldots,k-1$ we can write this expression (up to a sign) as
\ben
\left(\eta_{k-1} \cdots \eta_2 \eta _1 \sum_{\alpha=1}^{k-1} \eta^\alpha \prod_{i=1}^d \left(\sum_{\alpha = 1}^{k-1} \d \wbar_{i}^\alpha\right) \right) \prod_{\alpha=1}^{k-1} \prod_{i=1}^d \d \wbar_i^\alpha 
\een
The expression inside the parentheses is zero since each term in the sum over $\alpha$ involves a term like $\eta^\beta \eta^\beta = 0$. 
This completes the proof for $k=d$. 
\end{proof}

We now move on to the analytic part of the argument, where we show that for wheels with sufficiently large number of incoming edges, the analytic weight vanishes in the limit $\epsilon \to 0$. 
We point out that the method of proof of this lemma is nearly the same as the proof of Lemma 7.2.1 in \cite{bcov} in the context of BCOV theory on the flat odd Calabi-Yau manifold $\CC^d$, $d$ odd. 
We show here that the argument works in general for any holomorphic theory on flat space.

\begin{lem}\label{lem: tech 1}
Let $(k, (\vec{n}))$ be a pair of wheel data such that $k > d$.
Then the $\epsilon \to 0$ limit of the analytic weight
\ben
\lim_{\epsilon \to 0} W_{\epsilon < L}^{k, (n)}
\een
exists as a distribution on $\CC^{dk}$. 
\end{lem}

\begin{proof}
%\brian{add comment: “we will use integration by parts to bound the weight of equation (15) by an integral against a Gaussian. Using formulae for moments of a Gaussian, we will get a simple integral in the t parameters.”}

We will bound the absolute value of the weight in Equation (\ref{weight1}) and show that it has a well-defined $\epsilon\to 0$ limit.
First, consider the change of coordinates as in Equations (\ref{coords1}),(\ref{coords2}).
For any compactly supported function $\Phi$ we see that $W_{\epsilon < L}^{k, (n)} (\Phi)$ has the form
\be\label{weight2}
\begin{array}{lllllll}
\displaystyle \int_{w^k \in \CC^d} \d^{d} w^k \int_{(w_1,\ldots,w_{k-1}) \in (\CC^d)^{k-1}} & \displaystyle\left(\prod_{\alpha=1}^{k-1} \d^{d} w^\alpha\right) \Phi(w^1,\ldots,w^k) \left(\prod_{\alpha=1}^{k-1} \left(\frac{\partial}{\partial w^\alpha}\right)^{\vec{n}^\alpha}P^{an}_{\epsilon < L} (w^\alpha) \right) \\ & \displaystyle \times \sum_{\alpha=1}^{k-1} \left(\frac{\partial}{\partial w^\alpha}\right)^{\vec{n}^k} P^{an}_{\epsilon<L} \left(\sum_{\alpha=1}^{k-1} w^\alpha\right) .
\end{array}
\ee
For $\alpha = 1,\ldots,k-1$ the notation $P^{an}_{\epsilon < L} (w^\alpha)$ makes sense since $P^{an}_{\epsilon<L}(z^\alpha,z^{\alpha+1})$ is only a function of $w^\alpha = z^{\alpha+1}-z^\alpha$.
Similarly $P^{an}_{\epsilon<L}(z^{k+1},z^1)$ is a function of 
\ben
z^k - z^1 = \sum_{\alpha=1}^{k-1} w^\alpha . 
\een
Expanding out the propagators the weight takes the form
\ben
\begin{array}{lll}
& \displaystyle \int_{w^k \in \CC^d} \d^{2d} w^k \int_{(w_1,\ldots,w_{k-1}) \in (\CC^d)^{k-1}} \left(\prod_{\alpha=1}^{k-1} \d^{2d} w^\alpha\right) \Phi(w^1,\ldots,w^k) \int_{(t_1,\ldots,t_k) \in [\epsilon,L]^k} \prod_{\alpha=1}^k \frac{\d t_\alpha}{(4 \pi t_\alpha)^d} \\
& \displaystyle \times \sum_{i_1,\ldots,i_{k-1} =1}^d \epsilon_{i_1\cdots,i_k} \left(\frac{\wbar^1_{i_1}}{4t_1} \frac{(\wbar^1)^{n^1}}{4t^{|n^1|}}\right) \cdots \left(\frac{\wbar^{k-1}_{i_{k-1}}}{4t_{k-1}}\frac{(\wbar^{k-1})^{n^{k-1}}}{4t^{|n^{k-1}|}}\right) \left(\sum_{\alpha=1}^{k-1} \frac{\wbar^\alpha_{i_k}}{4t_k} \cdot \frac{1}{t^{|n^k|}} \left(\sum_{\alpha=1}^{k-1} \wbar^\alpha\right)^{n^k}\right) \\
& \displaystyle \times \exp\left(- \sum_{\alpha=1}^{k-1} \frac{|w^{\alpha}|^2}{4t_\alpha} - \frac{1}{4t_k} \left|\sum_{\alpha=1}^{k-1} w^\alpha \right|^2\right)
\end{array}
\een
The notation used above warrants some explanation. 
Recall, for each $\alpha$ the vector of integers is defined as $n^\alpha = (n^{\alpha}_1,\ldots,n^{\alpha}_d)$. 
We use the notation
\ben
(\wbar^\alpha)^{n^\alpha} = \wbar^{n^\alpha_1}_1 \cdots \wbar^{n^\alpha_d}_d .
\een
Furthermore, $|n^\alpha| = n_1^\alpha + \cdots + n_d^\alpha$. 
Each factor of the form $\frac{\wbar^\alpha_{i_\alpha}}{t_\alpha}$ comes from the application of the operator $\frac{\partial}{\partial z_i}$ in $\dbar^*$ applied to the propagator. 
The factor $\frac{(\wbar^\alpha)^{n^\alpha}}{t^{|n^\alpha|}}$ comes from applying the operator $\left(\frac{\partial}{\partial w}\right)^{n^\alpha}$ to the propagator. 
Note that $\dbar^*$ commutes with any translation invariant holomorphic differential operator, so it doesn't matter which order we do this.

To bound this integral we will recognize each of the factors
\ben
\frac{\wbar^\alpha_{i_\alpha}}{4t_\alpha} \frac{(\wbar^\alpha)^{n^\alpha}}{4t^{|n^\alpha|}}
\een
as coming from the application of a certain holomorphic differential operator to the exponential in the last line.
We will then integrate by parts to obtain a simple Gaussian integral which will give us the necessary bounds in the $t$-variables. 
Let us denote this Gaussian factor by
\ben
E(w,t) := \exp\left(- \sum_{\alpha=1}^{k-1} \frac{|w^{\alpha}|^2}{4t_\alpha} - \frac{1}{4t_k} \left|\sum_{\alpha=1}^{k-1} w^\alpha \right|^2\right)
\een

For each $\alpha,i_\alpha$ introduce the $t=(t_1,\ldots,t_k)$-dependent holomorphic differential operator
\ben
D_{\alpha, i_\alpha}(t) := \left(\frac{\partial}{\partial w^\alpha_{i_\alpha}} - \sum_{\beta = 1}^{k-1} \frac{t_\beta}{t_1+\cdots + t_k} \frac{\partial}{\partial w_{i_\alpha}^{\beta}}\right)
\prod_{j=1}^d \left(\frac{\partial}{\partial w_j^\alpha} - \sum_{\beta =1}^{k-1} \frac{t_\beta}{t_1+\cdots + t_k} \frac{\partial}{\partial w_{j}^\beta}\right)
^{n_j^\alpha} .
\een
\begin{rmk}
This operator, and method to bound the integral, has appeared in the proof of Lemma 7.3.1 of \cite{bcov} and is motivated by a trick the author has learned from Si Li, first executed in \cite{LiFeynman} in the context of $2$-dimensional chiral theories on elliptic curves.
We propose that this method be referred to as “Si’s trick” for reducing the divergence of Feynman integrals of this type.
We have further generalized this method in \cite{GWcs} to a wider class of quantum field theories. 
\end{rmk}
The following lemma is an immediate calculation
\begin{lem}\label{lem: diff applied E}
One has
\ben
D_{\alpha,i_\alpha} E(w,t) = \frac{\wbar^\alpha_{i_\alpha}}{4t_\alpha} \frac{(\wbar^\alpha)^{n^\alpha}}{t^{|n^\alpha|}} E(w,t) . 
\een
\end{lem}

Note that all of the $D_{\alpha,i_{\alpha}}$ operators mutually commute. 
Thus, we can integrate by parts iteratively to obtain the following expression for the weight:
\ben
\begin{array}{lll}
& \displaystyle \pm \int_{w^k \in \CC^d} \d^{2d} w^k \int_{(w_1,\ldots,w_{k-1}) \in (\CC^d)^{k-1}}\left(\prod_{\alpha=1}^{k-1} \d^{2d} w^\alpha\right) \int_{(t_1,\ldots,t_k) \in [\epsilon,L]^k} \prod_{\alpha=1}^k \frac{\d t_\alpha}{(4 \pi t_\alpha)^d}  \\ 
& \displaystyle \times\left( \sum_{i_1,\ldots, i_k} \epsilon_{i_1\cdots,i_d} D_{1, i_1} \cdots D_{k-1,i_{k-1}} \sum_{\alpha=1}^{k-1} D_{\alpha, i_k} \Phi(w^1,\ldots,w^k) \right) \times \exp\left(- \sum_{\alpha=1}^{k-1} \frac{|w^{\alpha}|^2}{t_\alpha} - \frac{1}{t_k} \left|\sum_{\alpha=1}^{k-1} w^\alpha \right|^2\right) .
\end{array}
\een
%\brian{all the differential operators $D_{\alpha, i_\alpha}$ are uniformly bounded in $t$. To make these precise I should find what the uniform bound is.}

Since the operators $D_{i,j}$ are uniformly bounded in the variables $t_1, \ldots, t_k$, the absolute value of the weight is bounded by 
\be\label{weight bound1}
\begin{array}{lllll}
\displaystyle |W_{\epsilon < L}^{k, (n)}(\Phi)| \leq C \int_{w^k \in \CC^d} \d^{2d} w^k & \displaystyle \int_{(w^1,\ldots,w^{k-1}} \prod_{\alpha=1}^{k-1} \d^{2d} w^\alpha \Psi(w^1,\ldots,w^{k-1},w^k) \\ & \displaystyle \times  \int_{(t_1,\ldots,t_k) \in [\epsilon,L]^k} \d t_1 \ldots \d t_k \frac{1}{(4\pi)^{dk}} \frac{1}{t^d_1\cdots t^d_k} \times E(w,t)
\end{array}
\ee
where $\Psi$ is some compactly supported functnio on $\CC^{dk}$ that is independent of $t$. 

To compute the right hand side we will perform a Gaussian integration with respect to the variables $(w^1,\ldots,w^{k-1})$. 
To this end, notice that the exponential can be written as
\ben
E(w,t) = \exp\left(-\frac{1}{4} M_{\alpha\beta} (w^\alpha, w^\beta)\right)
\een
where $(M_{\alpha\beta})$ is the $(k-1)\times (k-1)$ matrix given by
\ben
\begin{pmatrix}
a_1 & b & b & \cdots & b \\
b & a_2 & b & \cdots & b \\
b & b & a_3 & \cdots & b \\
\vdots & \vdots & \vdots &  \ddots & \vdots \\
b & b & b & \cdots & a_{k-1}
\end{pmatrix} 
\een
where $a_\alpha = t_\alpha^{-1} + t_k^{-1}$ and $b = t_k^{-1}$.
The pairing $(w^{\alpha}, w^{\beta})$ is the usual Hermitian pairing on $\CC^d$, $(w^{\alpha}, w^{\beta}) = \sum_i w^{\alpha}_i \wbar^\beta_i$.
After some straightforward linear algebra we find that 
\ben
\det(M_{\alpha\beta})^{-1} = \frac{t_1\cdots t_k}{t_1+\cdots+t_k} .
\een 
We now perform a Wick expansion for the Gaussian integral in the variables $(w^1,\ldots,w^{k-1})$.
For a reference similar to the notation used here see the Appendix of our work in \cite{EWY}.
The inequality in (\ref{weight bound1}) becomes

\begin{align}\label{weight bound2}
|W_{\epsilon < L}^{k, (n)}(\Phi)| & \leq C' \int_{w^k \in \CC^d} \d^{2d} w^k \Psi(0, \ldots, 0, w^k) \int_{(t_1,\ldots,t_k) \in [\epsilon,L]^k} \d t_1 \ldots \d t_k \frac{1}{(4\pi)^{dk}} \frac{1}{(t_1\cdots t_k)^d}\left(\frac{t_1\cdots t_k}{t_1+\cdots+t_k}\right)^d + O(\epsilon) \\ & = C' \int_{w^k \in \CC^d} \d^{2d} w^k \Psi(0, \ldots, 0, w^k) \int_{(t_1,\ldots,t_k) \in [\epsilon,L]^k} \d t_1 \ldots \d t_k \frac{1}{(4\pi)^{dk}} \frac{1}{(t_1+\cdots+t_k)^d} + O(\epsilon) .
\end{align}
The first term in the Wick expansion is written out explicitly. 
The $O(\epsilon)$ refers to higher terms in the Wick expansion, which one can show all have order $\epsilon$, so disappear in the $\epsilon \to 0$ limit.
The expression $\Psi(0, \ldots, 0, w^k)$ means that we have evaluate the function $\Psi(w^1,\ldots, w^k)$ at $w^1=\ldots=w^{k-1} =0$ leaving it as a function only of $w^k$. 
In the original coordinates this is equivalent to setting $z^1=\cdots=z^{k-1} = z^k$.

Our goal is to show that $\epsilon \to 0$ limit of the right-hand side exists. 
The only $\epsilon$ dependence on the right hand side of (\ref{weight bound2}) is in the integral over the regulation parameters $t_1,\ldots, t_k$. 
Thus, it suffices to show that the $\epsilon \to 0$ limit of 
\ben
\int_{(t_1,\ldots,t_k) \in [\epsilon,L]^k} \frac{\d t_1 \ldots \d t_k}{(t_1+\cdots+t_k)^d}
\een
exists.
By the AM/GM inequality we have $(t_1+\cdots+t_k)^d \geq (t_1 \cdots t_d)^{d/k}$. 
So, the integral is bounded by
\ben
\int_{(t_1,\ldots,t_k) \in [\epsilon,L]^k}\frac{\d t_1 \ldots \d t_k}{(t_1+\cdots+t_k)^d} \leq \int_{(t_1,\ldots,t_k) \in [\epsilon,L]^k}\frac{\d t_1 \ldots \d t_k }{(t_1 \cdots t_k)^{d/k}} = \frac{1}{(1-d/k)^k} \left(\epsilon^{1-d/k} - L^{1-d/k}\right)^k .
\een
By assumption, $d < k$, so the right hand side has a well-defined $\epsilon \to 0$ limit. 
This concludes the proof.

%Now, since $t_\alpha / \sum_\beta t_\beta < 1$ for each $\alpha$ we have the following bound for the operators $D_{\alpha, i_\alpha}$:
%\bestar
%\left|D_{\alpha,i_{\alpha}} \Phi\right| & \leq \left(\left|\frac{\partial}{\partial w^\alpha_{i_\alpha}} \Phi\right| +  \sum_{\beta = 1}^{k-1}\frac{\partial}{\partial w_{i_\beta}^{\beta}}\right)
%\prod_{j=1}^d \left(\frac{\partial}{\partial w_j^\alpha} - \frac{1}{k} \sum_{\beta =1}^{k-1} \frac{\partial}{\partial w_{j}^\beta}\right)
%^{n_j^\alpha} \right| 
\end{proof}

\section{Chiral anomalies in arbitrary dimensions}

Renormalization is an important step in constructing a quantum field theory.
In the context of gauge theory, however, a consistent quantization requires that this renormalization behaves appropriately with respect to gauge symmetries present in the classical theory. 
This formalism for studying quantizations of gauge theories is due to Batalin-Vilkovisky \cite{BV}, and has been made mathematically rigorous in the work of Costello \cite{CosRenorm}.
The precise consistency of gauge symmetry with renormalization is encoded by the {\em quantum master equation}. 
Heuristically, one can think of the quantum master equation as a closedness condition on the path integral measure defined by the quantum action functional. 

The key idea is the following: once a classical theory has been renormalized, so that we have a $\hbar$-linear effective family of functionals $\{I[L]\}$ whose $\hbar = 0$ limit is the classical action, the next step to constructing a quantization is to solve the quantum master equation (QME) for each functional $I[L]$. 
(In fact, once the QME holds at a single positive length $L>0$ it holds for every other length by RG flow.)
Often, the QME is not satisfied by the functional $I[L]$, but there exists a ``correction" to $I[L]$ that does satisfy the QME. 
On the other hand, there may be unavoidable obstructions to solving this quantum master equation.
These are known as {\em anomalies} in the physics literature. 
Since our method for solving the QME is deformation-theoretic in nature, these anomalies appear as cohomology classes in the cochain complex of local functionals. 

In general, it is difficult to characterize such anomalies, but in the case of holomorphic theories on $\CC^d$ our result of one-loop finiteness from the previous section makes this problem much more tractable. 
Indeed, since there are no counterterms required, we can plug in the RG flow of the classical action functional and study the quantum master equation directly. 
As is usual in perturbation theory, one works order by order in $\hbar$ to construct a quantization.
However, in this section we continue to work linearly in $\hbar$, which is to say we study solutions to the quantum master equation modulo $\hbar^2$. 

\subsection{The quantum master equation}

In the BV formalism, as developed in \cite{CosRenorm,CG1,CG2}, one has the following definition of a quantum field theory.

\begin{dfn}\label{dfn: qft}
A {\em quantum field theory} in the BV formalism consists of a free BV theory $(\sE, Q, \omega)$ and an effective family of functionals
\ben
\{I[L]\}_{L \in (0,\infty)} \subset \sO^+_{P,sm}(\sE)[[\hbar]]
\een
that satisfy:
\begin{enumerate}[(a)]
\item the exact renormalization group (RG) flow equation
\ben
I[L'] = W(P_{L<L'}, I[L]);
\een
\item the scale $L$ quantum master equation (QME) at every length scale $L$:
\ben
(Q + \hbar \Delta_L) e^{I[L]/\hbar} = 0.
\een
\item as $L \to 0$, the functional $S[L]$ has an asymptotic expansion that is local.
\end{enumerate}
\end{dfn}

The first part of the definition, namely RG flow, was the phenomena we studied in the previous section. 
We turn our attention to part two of the definition of a QFT. 
The regularized quantum master equation at scale $L$ can equivalently be written as
\ben
Q I[L] + \hbar \Delta_L I[L] + \frac{1}{2} \{I[L], I[L]\}_L = 0 ;
\een
Combined with part (c), the $\hbar \to 0$, $L \to 0$ limit of the above equation is precisely the classical master equation for the local functional $\lim_{L \to 0} I[L] \mod \hbar$. 
A quantization of a classical functional $I \in \oloc(\sE)$ is a quantization $\{I[L]\}$ as above whose $\hbar \to 0$ limit agrees with $I$. 

In general, not every classical interaction admits a quantization.
The {\em obstruction} to satisfying the quantum master equation order by order in $\hbar$ is given by the following inductive definition.

\begin{dfn}
Suppose $I[L] \in \sO(\sE)[[\hbar]] / \hbar^{n+1}$ solves the QME modulo $\hbar^{n}$. 
The scale $L$ obstruction to solving the QME modulo $\hbar^{n}$ is 
\ben
\Theta^{(n)} [L] = \hbar^{-n} \left(Q I [L] + \hbar \Delta_L I[L] + \frac{1}{2} \{I[L], I[L]\}_L \right) \in \sO(\sE) .
\een
\end{dfn}

Equivalently, we can write the obstruction as $\Theta^{(n)} [L] = \hbar^{-n+1} e^{-I[L]/\hbar} (Q + \hbar\Delta_L)e^{I[L]/\hbar}$. 

As a consequence of part (c) in the definition of a QFT above, the $L \to 0$ limit of the obstruction is defined and determines a cohomological degree $+1$ local functional
\ben
\Theta^{(n)} = \lim_{L \to 0} \Theta^{(n)}[L] \in \oloc(\sE) .
\een
Moreover, $\Theta^{(n)}$ is closed for the differential $Q + \{I,-\}$, where $I = \lim_{L \to 0} I[L] \mod \hbar \in \oloc(\sE)$. 

In the remainder of this section we return to the holomorphic setting. 
Fix a classical holomorphic theory $(V, Q^{hol}, (-,-)_V, I^{hol})$ on $\CC^d$.
As usual, denote by $\sE_V = (\Omega^{0,*}(\CC^d, V), \dbar + Q^{hol})$ the linearized BRST complex of fields and $I = \int I^{\Omega^{0,*}}$ the classical interaction.
Let $I[L] = \lim_{\epsilon \to 0} W(P_{\epsilon<L}, I) \mod \hbar^2$ be the one-loop renormalization group flow using the propagator defined in Section \ref{sec: hol renorm}.

\subsection{The QME for holomorphic theories}

The main result of this section is a characterization of the one-loop obstruction for holomorphic theories. 
Before jumping into the calculation, we state the following lemma, which is a simplification of the QME given our assumptions.
Note that we only study one-loop effects here. 

The core idea of this lemma has appeared in the Appendix of \cite{LiLi}, where they study the one-loop anomaly for a general perturbative field theory (with counterterms potentially present).
We repeat the result here since the statement and proof simplifies in the context of holomorphic field theory.

\begin{lem}
Let $\Theta[L] = \Theta^{(1)}[L]$ be the one-loop obstruction to the QME at scale $L$.
Then, one has
\be\label{anomaly lem}
\hbar \Theta[L] = Q^{hol} I[L] + \frac{1}{2} \lim_{\epsilon \to 0} e^{-I/\hbar} e^{\hbar \partial_{P_{\epsilon < L}}} \left(\{I,I\}_\epsilon e^{I/\hbar}\right) \mod \hbar^2 .
\ee
\end{lem}

\begin{proof}

We write the obstruction as $\Theta[L] = e^{-I[L]/\hbar} (Q + \hbar\Delta_L)e^{I[L]/\hbar}$.
Notice that formally the one loop RG flow can equivalently be written as $e^{I[L]/\hbar} = \lim_{\epsilon \to 0} e^{\hbar \partial_{P_{\epsilon<L}}} e^{I / \hbar} \mod \hbar^2$.

Applying the operator $Q + \hbar \Delta_L$ to both sides, we obtain
\ben
(Q + \hbar \Delta_L) e^{I[L]/\hbar} = \lim_{\epsilon \to 0} (Q + \hbar \Delta_L)  \left(e^{\hbar \partial_{P_{\epsilon<L}}} e^{I / \hbar}\right) .
\een

The operator $Q$ satisfies $[Q, \partial_{P_{\epsilon<L}}] = \Delta_L - \Delta_\epsilon$. 
So, acting on functionals one has $(Q + \hbar \Delta_L) e^{\hbar \partial_{P_{\epsilon<L}}} = e^{\hbar \partial_{P_{\epsilon<L}}} (Q + \hbar \Delta_\epsilon)$. 
The above then simplifies to
\ben
\lim_{\epsilon \to 0} (Q + \hbar \Delta_L)  \left(e^{\hbar \partial_{P_{\epsilon<L}}} e^{I / \hbar}\right) = \lim_{\epsilon \to 0}  e^{\hbar \partial_{P_{\epsilon < L}}} (Q + \hbar \Delta_\epsilon) e^{I / \hbar} .
\een
Since $\Delta_\epsilon$ is a BV operator with respect to the bracket $\{-,-\}_{\epsilon}$, we can rewrite the right-hand side as
\ben
\frac{1}{\hbar} \lim_{\epsilon \to 0} e^{\hbar P_{\epsilon < L}} (Q I + \hbar \Delta_\epsilon I + \frac{1}{2}\{I, I\}_\epsilon) e^{I / \hbar}.
\een

For every $\epsilon > 0$ we have $\Delta_\epsilon I = 0$.
This is because $I$ is a local functional and $\Delta_{\epsilon}$ involves contraction with a factor of $\prod (\d \zbar_i - \d \wbar_i)$ which vanishes along the diagonal.
Moreover, since $I$ comes from a holomorphic Lagrangian we have $\dbar I = 0$.

Thus, the only terms remaining inside the parantheses in the above expression are $Q^{hol} I + \frac{1}{2} \{I,I\}_{\epsilon}$. 
We conclude that the obstruction $\Theta[L]$ can be expressed as

\begin{align*}
\Theta[L] & = \frac{1}{\hbar} \lim_{\epsilon \to 0} e^{-I/\hbar} e^{\hbar \partial_{P_{\epsilon < L}}} \left(Q^{hol} I + \frac{1}{2} \{I,I\}_\epsilon e^{I/\hbar}\right) \mod \hbar^2 \\
& = \frac{1}{\hbar} Q^{hol} I[L] + \frac{1}{2\hbar} \lim_{\epsilon \to 0} e^{-I/\hbar} e^{\hbar \partial_{P_{\epsilon <L}}} \left(\{I,I\}_\epsilon e^{I/\hbar}\right) \mod \hbar^2
\end{align*}
as desired.
In the second line, we have again used the fact that the operators $Q^{hol}$ and $\partial_{P_{\epsilon <L}}$ commute.
\end{proof}

As we saw above, the anomaly $\Theta[L]$ has a well-defined $L \to 0$ limit as a local functional and it is closed for the classical differential.
Before stating the result, we need a modification of the definition of the weight of a given Feynman diagram. 
If $\Gamma$ is a graph with a distinguished edge $e$, let $W_{\Gamma,e}(P_{\epsilon<L},K_{\epsilon}, I)$ denote the weight of the graph as defined before, except with one minor difference.
Instead of placing $P_{\epsilon <L}$ at each internal edge, we place $K_\epsilon$ at the edge labeled $e$ and $P_{\epsilon<L}$ on the remaining edges.
The main result of this section is the following.

\begin{prop}\label{lem: chiral anomaly}
The obstruction $\Theta = \lim_{L \to 0} \Theta[L] \in \oloc(\sE_V)$ to satisfying the one-loop quantum master equation is given by the expression
\be\label{anomaly}
\hbar \Theta = Q^{hol} \lim_{L \to 0} I[L] + \frac{1}{2} \lim_{L \to 0} \lim_{\epsilon \to 0} \sum_{\Gamma \in {\rm Wheel}_{d+1}, e} W_\Gamma(P_{\epsilon < L}, K_\epsilon,I)
\ee
where the sum is over all wheels with $(d+1)$-vertices and distinguished edges thereof.
In particular, when $Q^{hol} = 0$ (so that the first term vanishes), the anomaly is expressed as the sum over wheels with exactly $(d+1)$-vertices. 
\end{prop}

\begin{rmk}
This result, and Lemma \ref{lem: anomalyanalysis} below, can be seen as generalizations of results that have appeared in the work \cite{bcov}, specifically Lemma 7.2.7 and its proof, where the one-loop anomaly present in holomorphic Chern-Simons on the flat odd Calabi-Yau manifold $\CC^d$ is computed.
The method of our proof is similar to the method employed there.
\end{rmk}

This obstruction determines an element in the cohomology of the local deformation complex
\ben
[\Theta] \in H^1\left(\oloc(\sE_V), \dbar + Q^{hol} + \{I,-\}\right) .
\een
This is a complete characterization of the cohomological obstruction to satisfying the quantum master equation for the classical theory $I$. 
If we chose any other quantization of $\{I'[L]\}$ of $I$, say coming from a different gauge fixing condition, we obtain class cohomologous to this $[\Theta] = [\Theta']$.

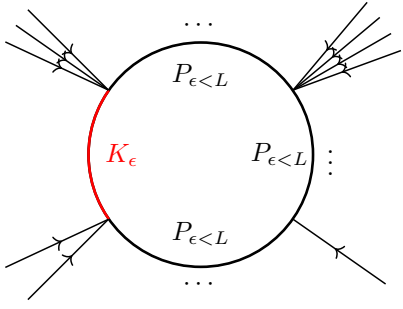
\begin{figure}
\begin{center}
\begin{tikzpicture}[line width=.2mm, scale=1.5]

%\pgfmathsetmacro{\ex}{0}
%\pgfmathsetmacro{\ey}{1}

%\draw (\ex,\ey) ++(45:.8) arc (45:-45:.8);

		\draw[fill=black] (0,0) circle (1cm);
		%\draw[fill=red] (0,0) arc (145:215:1);
		\draw[fill=white] (0,0) circle (0.99cm);
		\draw[line width=0.35mm,red] ++(145:0.995) arc (145:215:0.995);
		%\draw[red] (0,0) arc (30:60:3);

		\draw[fermion](140:2) -- (145:1);
		\draw[fermion](145:2) -- (145:1);
		\draw[fermion](150:2) -- (145:1);
			%\node at (145:0.85) {$v_0$};
			%\node at (145:2.25) {$\gamma$};
		\node at (90:1.15) {$\cdots$};
		%\node at (90:0.8) {$P_{\epsilon<L}$};
		\draw[fermion](210:2) -- (215:1cm);
		\draw[fermion](220:2) -- (215:1cm);
			%\node at (215:0.85) {$v_{d}$};
			%\node at (215:2.25) {$\gamma$};
		\node[red] at (180:0.7) {$K_\epsilon$};
		\draw[fermion](27.5:2) -- (35:1);
		\draw[fermion](32.5:2) -- (35:1);
		\draw[fermion](37.5:2) -- (35:1);
		\draw[fermion](42.5:2) -- (35:1);
			%\node at (35:0.85) {$v_{\alpha}$};
			%\node at (35:2.25) {$\gamma$};
		\draw[fermion](-35:2) -- (-35:1cm);
			%\node at (-35:0.85) {$v_{\beta}$};
			%\node at (-35:2.25) {$\gamma$};
		\node at (0:1.15) {$\vdots$};
		\node at (0:0.7) {$P_{\epsilon<L}$};
		\node at (270:1.15) {$\cdots$};
		\node at (270:0.7) {$P_{\epsilon<L}$};
		\node at (90:0.7) {$P_{\epsilon<L}$};
	    	\clip (0,0) circle (1cm);
\end{tikzpicture}
\caption{The second term in Equation (\ref{anomaly}) representing the holomorphic anomaly.}
\label{fig:chernwheel}
\end{center}
\end{figure}

\begin{proof}[Proof of Proposition \ref{lem: chiral anomaly}]

Like the proof of the non-existence of counterterms for holomorphic theories, the proof of this result will be the consequence of an explicit calculations and bounds of certain Feynman diagrams. 

Note that the first term, involving $Q^{hol}$, is the $L \to 0$ limit of the right-hand side of Equation (\ref{anomaly lem}). 
Thus, it suffices to focus on the second term. 

We express the quantity
\be\label{graph exp}
\lim_{\epsilon \to 0} e^{-I/\hbar} e^{\hbar \partial_{P_{\epsilon < L}}} \left(\{I,I\}_\epsilon e^{I/\hbar}\right)\mod \hbar^2
\ee
as a sum over graphs.
By assumption, we are only looking at graphs of genus one which look like wheels with possible trees attach.
Graphically, the quantity $\{I,I\}_\epsilon$ is the graph of two vertices with a separating edge labeled by the heat kernel $K_\epsilon$.
Thus, all weights appearing in the expansion of (\ref{graph exp}) attach the propagator $P_{\epsilon<L}$ to all edges besides a single distinguished edge $e$, which is labeled by $K_\epsilon$. 
Thus, as a over a sum of graphs, we see that the following two types of weights occur in the expansion of (\ref{graph exp}).
\begin{enumerate}[(a)]
\item the distinguished edge $e$ is separating;
\item the distinguished edge $e$ is {\em not} separating, and so appears as the internal edge of the wheel portion of the graph.
\end{enumerate}

The classical master equation implies that the $\epsilon \to 0, L \to 0$ limit of weights of Type (a) go to zero.
Thus, we must only consider the weights of Type (b). 

The result will follow from two steps.
These should seem familiar from the proof of the main result about the existence of no counterterms.
\begin{enumerate}
\item If $\Gamma$ is a wheel with $k < d+1$ vertices, then $ W_\Gamma(P_{\epsilon < L}, K_\epsilon,I)
 = 0$ identically. 
\item If $\Gamma$ is a wheel with $k > d+1$ vertices, then $\lim_{\epsilon \to 0} W_\Gamma(P_{\epsilon < L}, K_\epsilon,I) =0$.
\end{enumerate} 

The proof of both of these facts is only dependent on the analytic part of the weights. 
Thus, it suffices to make the same reduction as we did in the previous section.
To extract that analytic part of the graph we proceed as in Definition \ref{dfn: analytic weight}.
If $(k, (\vec{n}))$ is a pair of wheel data (recall $k$ labels the number of vertices and $\vec{n}$ labels the derivatives at each vertex) define the smooth distribution
\ben
\Tilde{W}_{\epsilon < L}^{k, (n)} : C_c^\infty((\CC^d)^k) \to \CC,
\een
that sends a smooth compactly supported function $\Phi \in C_c^\infty((\CC^d)^k) = C_c^\infty(\CC^{dk})$ to
\be\label{anomalyweight1}
\Tilde{W}_{\epsilon < L}^{k, (n)} (\Phi) = \int_{(z^1,\ldots, z^k) \in (\CC^d)^k} \prod_{\alpha=1}^k \d^d z^\alpha \Phi(z^1,\ldots,z^\alpha)  \left(\frac{\partial}{\partial z^k}\right)^{\vec{n}^k} K_{\epsilon}(z^1,z^k) \prod_{\alpha = 1}^{k-1} \left(\frac{\partial}{\partial z^\alpha}\right)^{\vec{n}^\alpha} P_{\epsilon < L}^{an}(z^\alpha, z^{\alpha+1}) .
\ee

Item (1) follows from the following observation.

\begin{lem}\label{lem: anomalyanalysis} Let $(k, (\vec{n}))$ be a pair of wheel data.
If the number of vertices $k$ satisfies $k \leq d$ then
\ben
\Tilde{W}_{\epsilon < L}^{k, (n)}  = 0
\een
as a distribution on $\CC^{dk}$ for any $\epsilon,L > 0$. 
\end{lem}
\begin{proof}
In fact, the integrand of (\ref{anomalyweight1}) is identically zero provided $k \leq d$ by a simple observation of the differential form type.
Consider the factor in the integrand of $\Tilde{W}_{\epsilon < L}^{k, (n)}$ given by
\ben
 \left(\frac{\partial}{\partial z^k}\right)^{\vec{n}^k} K_{\epsilon}(z^1,z^k) \prod_{\alpha = 1}^{k-1} \left(\frac{\partial}{\partial z^\alpha}\right)^{\vec{n}^\alpha} P_{\epsilon < L}^{an}(z^\alpha, z^{\alpha+1}).
\een
Making the usual change of coordinates $w^\alpha = z^{\alpha +1} - z^\alpha$ and $w^k = z^k$ we see that this factor is proportional to the following constant coefficient differential form
\ben
\left(\prod_{i=1}^d \left(\sum_{\alpha = 1}^{k-1} \d \wbar_{i}^\alpha\right) \right) \prod_{\alpha=1}^{k-1}\left( \eta^\alpha \prod_{i=1}^d \d \wbar_i^\alpha\right) .
\een
Note that this differential form only involves the coordinates $(w_i^\alpha)$ for $\alpha = 1,\ldots,k-1$.
Thus, we may consider it as a Dolbeualt form on $\CC^{d(k-1)}$.
As such, it is of the type $(0, d + (k-1)(d-1)) = (0, (d-k+1) + d(k-1))$. 
Clearly, $(d-k+1) + d (k-1) > d (k-1)$ provided $k \leq d$.
Thus, the weight is identically zero provided $k \leq d$, as desired.
\end{proof}

Item (2) follows from the following technical lemma that the analytic weight associated to the wheels of valency $k > d+1$ vanish in the limit $\epsilon \to 0$. 

\begin{lem}
Let $(k, (\vec{n}))$ be a pair of wheel data such that $k > d+1$.
Then the $\epsilon \to 0$ limit of the analytic weight
\ben
\lim_{\epsilon \to 0} \Tilde{W}_{\epsilon < L}^{k, (n)} = 0
\een
is identically zero as a distribution on $\CC^{dk}$. 
\end{lem}
\begin{proof}
The proof is very similar to the argument we gave in the proof of Lemma \ref{lem: tech 1}, so we will be a bit more concise.
First, we make the familiar change of coordinates as in Equations (\ref{coords1}),(\ref{coords2}).
Using the explicit form the heat kernel and propagator we see that for any $\Phi \in C^\infty_c(\CC^{dk})$ the weight is
\ben
\begin{array}{lll}
\Tilde{W}_{\epsilon < L}^{k, (n)}(\Phi) & = \displaystyle \int_{w^k \in \CC^d} \d^{2d} w^k \int_{(w_1,\ldots,w_{k-1}) \in (\CC^d)^{k-1}} \left(\prod_{\alpha=1}^{k-1} \d^{2d} w^\alpha\right) \Phi(w^1,\ldots,w^k) \\ & \displaystyle \times \int_{(t_1,\ldots,t_k) \in [\epsilon,L]^{k-1}} \frac{1}{(4 \pi \epsilon)^d} \prod_{\alpha=1}^{k-1} \frac{\d t_\alpha}{(4 \pi t_\alpha)^d} \\
& \displaystyle \times \sum_{i_1,\ldots,i_{k-1} =1}^d \epsilon_{i_1,\ldots,i_d} \left(\frac{\wbar^1_{i_1}}{t_1} \frac{(\wbar^1)^{n^1}}{4t^{|n^1|}}\right) \cdots \left(\frac{\wbar^{k-1}_{i_{k-1}}}{4t_{k-1}}\frac{(\wbar^{k-1})^{n^{k-1}}}{4t^{|n^{k-1}|}}\right) \left(\frac{1}{4t^{|n^k|}} \left(\sum_{\alpha=1}^{k-1} \wbar^\alpha\right)^{n^k}\right) \\
& \displaystyle \times \exp\left(- \sum_{\alpha=1}^{k-1} \frac{|w^{\alpha}|^2}{4t_\alpha} - \frac{1}{4\epsilon} \left|\sum_{\alpha=1}^{k-1} w^\alpha \right|^2\right) .
\end{array}
\een
We will integrate by parts to eliminate the factors of $\wbar_i^\alpha$.

For each $1 \leq \alpha < k$ and $i_\alpha$, define the $\epsilon$ and $t=(t_1,\ldots,t_{k-1})$-dependent holomorphic differential operator
\ben
D_{\alpha, i_\alpha}(t) := \left(\frac{\partial}{\partial w^\alpha_{i_\alpha}} - \sum_{\beta = 1}^{k-1} \frac{t_\beta}{t_1+\cdots + t_{k-1} + \epsilon} \frac{\partial}{\partial w_{i_\alpha}^{\beta}}\right)
\prod_{j=1}^d \left(\frac{\partial}{\partial w_j^\alpha} - \sum_{\beta =1}^{k-1} \frac{t_\beta}{t_1+\cdots + t_{k-1} + \epsilon} \frac{\partial}{\partial w_{j}^\beta}\right)
^{n_j^\alpha} .
\een
And the $\epsilon,t$-dependent holomorphic differential operator
\ben
D_{k}(t) = \prod_{j=1}^d \left(\frac{\partial}{\partial w_j^k} - \sum_{\beta =1}^{k-1} \frac{t_\beta}{t_1+\cdots + t_{k-1} + \epsilon} \frac{\partial}{\partial w_{j}^\beta}\right)^{n_j^k} .
\een

By a completely analogous version of Lemma \label{lem: diff applied E} the operators above allow us to integrate by parts and express the weight in the form
\ben
\begin{array}{lll}
\Tilde{W}_{\epsilon < L}^{k, (n)}(\Phi) & = \displaystyle \pm \int_{w^k \in \CC^d} \d^{2d} w^k \int_{(w_1,\ldots,w_{k-1}) \in (\CC^d)^{k-1}}\left(\prod_{\alpha=1}^{k-1} \d^{2d} w^\alpha\right) \\ & \displaystyle \times \int_{(t_1,\ldots,t_{k-1}) \in [\epsilon,L]^{k-1}} \frac{1}{(4 \pi \epsilon)^d} \prod_{\alpha=1}^{k-1} \frac{\d t_\alpha}{(4 \pi t_\alpha)^d}  \\ 
& \displaystyle \times\left( \sum_{i_1,\ldots, i_{k-1}} \epsilon_{i_1\cdots,i_d} D_{1, i_1}(t) \cdots D_{k-1,i_{k-1}} (t) D_k(t) \Phi(w^1,\ldots,w^k) \right) \\ &\displaystyle \times \exp\left(- \sum_{\alpha=1}^{k-1} \frac{|w^{\alpha}|^2}{4t_\alpha} - \frac{1}{4\epsilon} \left|\sum_{\alpha=1}^{k-1} w^\alpha \right|^2\right) .
\end{array}
\een 
Observe that the operators $D_{\alpha,i_\alpha}(t), D_k(t)$ are uniformly bounded in $t$.
Thus, there exists a constant $C = C(\Phi) > 0$ depending only on the function $\Phi$ such that we can bound the weight as
\be\label{anomaly bound 1}
\begin{array}{lll}
|\Tilde{W}_{\epsilon < L}^{k, (n)}(\Phi)| & \leq & \displaystyle C \int_{(w^1,\ldots,w^{k-1}} \prod_{\alpha=1}^{k-1} \d^{2d} w^\alpha  \int_{(t_1,\ldots,t_{k-1}) \in [\epsilon,L]^{k-1}} \d t_1 \ldots \d t_k \frac{1}{\epsilon^d t^d_1\cdots t^d_{k-1}} \\
& & \displaystyle \times \exp\left(- \sum_{\alpha=1}^{k-1} \frac{|w^{\alpha}|^2}{4t_\alpha} - \frac{1}{4\epsilon} \left|\sum_{\alpha=1}^{k-1} w^\alpha \right|^2\right) .
\end{array}
\ee
Thus, to show that the limit $\lim_{L \to 0} \lim_{\epsilon \to 0}  \Tilde{W}_{\epsilon < L}^{k, (n)}(\Phi) = 0$ it suffices to show that the limit of the right-hand side vanishes. 

The Gaussian integral over the variables $w^\alpha_i$ contributes the following factor
\ben
\int_{(w^1,\ldots,w^{k-1}} \prod_{\alpha=1}^{k-1} \d^{2d} w^\alpha \exp\left(- \sum_{\alpha=1}^{k-1} \frac{|w^{\alpha}|^2}{4t_\alpha} - \frac{1}{4\epsilon} \left|\sum_{\alpha=1}^{k-1} w^\alpha \right|^2\right) = C' \left(\frac{\epsilon t_1 \cdots t_{k-1}}{\epsilon + t_1 + \cdots + t_{k-1}}\right)^{d} .
\een
Where $C'$ involves factors of $2$ and $\pi$.
Plugging this back in to the right-hand side of (\ref{anomaly bound 1}) we see that 
\ben
|\Tilde{W}_{\epsilon < L}^{k, (n)}(\Phi)| \leq C C' \int_{[\epsilon,L]^{k-1}} \frac{\d t_1 \cdots \d t_{k-1}}{(\epsilon + t_1 + \cdots + t_{k-1})^d} \leq C C' \prod_{\alpha = 1}^{k-1} \int_{t_\alpha = \epsilon}^L \d t_\alpha t_\alpha^{-d/(k-1)} .
\een
In the second inequality we have used the fact that $\epsilon > 0$ and the AM-GM inequality.  
It is immediate to see that the $\epsilon \to 0$ limit of the above exists provided $k > d+1$, which is the situation we are in, and that the $L \to 0$ limit vanishes. 
\end{proof}

This completes the proof of Proposition \ref{lem: chiral anomaly} 
\end{proof}
%By the compatibility between the QME and RG flow, it suffices to solve this equation at any scale $L$. 
%Thus, $(Q + \hbar \Delta_L)e^{I[L] / \hbar} = 0$ if and only if 
%\begin{lem}
%\ben
%\lim_{L \to 0} \lim_{\epsilon \to 0} e^{-I / \hbar} e^{\hbar \partial_{P_{\epsilon < L}}} \left(\{I, I\}_\epsilon e^{I / \hbar}\right) = \lim_{L \to 0} \lim_{\epsilon \to 0} \sum_{\Gamma} W_\Gamma(P_{\epsilon<L}, K_\epsilon, I) \mod \hbar^2 .
%\een
%\end{lem}

\subsubsection{Relation to the ABJ anomaly}

The lemma we have just proved implies that for holomorphic theories on $\CC^d$ the anomaly is given by evaluating a collection of wheel diagrams with exactly $d+1$ vertices. 
This expression for the obstruction fits into a generic class of of one-loop anomalies from gauge theory called the Adler-Bell-Jackiw (ABJ) anomaly \cite{Adler, BJ}.
This anomaly is most commonly associated with four dimensional gauge theory.

We recall the basic setup for the ABJ anomaly. 
Consider a free Dirac fermion $\Psi$ on $\RR^4$ coupled to a background gauge field $A \in \Omega^1(\RR^4) \tensor \fg$. 
For this to make sense, $\Psi$ is taken to be in valued in a representation $V$ of the Lie algebra $\fg$ so we may think of it as an element $\Psi \in \sS(\RR^4) \tensor V$.
Here, $\sS(\RR^4)$ is the space of sections of the full spinor bundle on $\RR^4$.
The action functional is
\ben
S(A, \Psi) = \int \<\Psi, \slashed\partial_A \Psi\>_V
\een
where $\slashed\partial_A = \slashed\partial + [A,-]$ is the $A$-coupled Dirac operator. 
We are implicitly using the canonical spin invariant symplectic pairing $\sS \tensor \sS \to \Omega^4(\RR^4) = C^\infty(\RR^4) \d^4 x$ and a $\fg$-invariant pairing $\<-,-\>_V : V \tensor V \to \CC$, to obtain a local functional.

For any smooth map $\alpha : \RR^4 \to \fg$, the infinitesimal transformation $\Psi \to \Psi + \epsilon [\alpha, \Psi]$ (where $\epsilon$ is an even parameter of square zero) is a classical symmetry of $S(A, \Psi)$. 
Quantum mechanically, there is a one-loop anomaly which measures the failure of the path integral to be invariant with respect to this symmetry. 
It is a well-known calculation, see for instance \cite{FujikawaSuzuki}, that this anomaly is measured by the following local functional
\be\label{ABJ}
\int \Tr_V \left(\alpha F_A F_A\right) .
\ee
The trace is taken in the representation $V$.
The fundamental calculation is the infamous "triangle diagram", 
%see Figure \ref{fig: triangle},
where two vertices are labeled by the gauge field and the third by $\alpha$.
In practice, physicists express the anomaly as a failure for the Noether current associated to the symmetry $\alpha$ to be divergenceless. 

There is the following holomorphic version of this anomaly. 
Again, let $V$ be a $\fg$ representation.
Consider the following action functional on $\CC^2$:
\ben
S(A, \beta,\gamma) = \int \<\beta, \dbar_A \gamma\>_V
\een
where $\gamma : \CC^2 \to V$, $\beta \in \Omega^{2,1}(\CC^2 , V)$, and $A \in \Omega^{0,1}(\CC^2, \fg)$. 
Since $A$ is a $(0,1)$ form it defines a deformation of the trivial holomorphic $G$-bundle. 
Although we have not put this theory in the BV formalism, there is a natural way to do so. 
The infinitesimal symmetry we contemplate is of the form $\gamma \to \gamma + \epsilon [\alpha, \gamma]$ where $\alpha : \CC^2 \to \fg$. 
We study the anomaly to quantizing this symmetry to one-loop.
Following the result for the anomaly given in the previous section, one sees that it is computed by a wheel with three vertices. 
For type reasons, one vertex is labeled $\alpha$ and the other two are labeled by the gauge fields $A$.
%A special case of a general calculation performed later in Chapter \ref{chap: symmetries} of this thesis computes the value of the diagram as
\ben
\int {\rm Tr}_V(\alpha \partial A \partial A) .
\een
This is the holomorphic version of ABJ anomaly (\ref{ABJ}). 
Note that there are no terms of order $A^3$ or above. 
In fact, the functional $\int {\rm Tr}(\alpha F_A F_A)$ is cohomologous to the expression above in the local deformation complex.

%\begin{rmk} 
%In the next section, using the concept of the "equivariant" BV formalism, we will make coupling background fields to a classical theory precise. 
%Then, Lemma \ref{lem: chiral anomaly} above applies rigorously to give the form of anomaly we have given. 
%We will see a precise statement of this for the holomorphic current algebra in our proof of the Grothendieck-Riemann-Roch theorem in Chapter \ref{chap: symmetries}.
%\end{rmk}

\begin{rmk} 
We have already shown how familiar topological theories can be cast in a holomorphic language.
For instance, topological $BF$ theory is a holomorphic deformation of holomorphic $BF$ theory. 
It is a peculiar consequence of the above result that such topological theories also admit a simple regularization procedure. 
Without much more difficulty, one can extend this to certain topological theories to odd dimensional manifolds of the form $X \times S$, where $X$ is a complex manifold and $S$ is a real one-dimensional manifold. 
We consider the theory as a product of a holomorphic theory on $X$ and a one-dimensional topological theory on $S$. 
This can be further extended to transversely holomorphic foliations \cite{THF1,THF2}, which we will study in a future publication.
Further, often topological BF theory further deforms to Yang-Mills.
It would be interesting to apply our analysis above to such gauge theories. 
\end{rmk}

%\subsection{Anomaly cancellation}

%The characterization 

%\section{Applications}

%\subsection{One-loop renormalization for topological theories}

%\subsection{Chiral anomalies in arbitrary dimensions}

\appendix

\section{Some functional analysis}

Homological algebra plays a paramount role in our approach to quantum field theory.
We immediately run into a subtle issue, which is that the underlying graded spaces of the complexes of fields we are interested in are infinite dimensional, so care must be taken when defining constructions such as duals and homomorphism spaces. 
A common approach to dealing with issues of infinite dimensional linear algebra is to consider vector spaces equipped with a topology. 
A problem with this is that the category of topological vector spaces is not an abelian category, so doing any homological algebra in this naive category is utterly hopeless. 
It is therefore advantageous to enlarge this to the category of {\em differentiable vector spaces}.
The details of this setup are carried out in the Appendix of \cite{CG1}, but we will recall some key points for completeness of exposition.
In this appendix we also set up our notation for duals and function spaces. 

Let ${\rm Mfld}$ be the site of smooth manifolds.
The covers defining the Grothendieck topology are given by surjective local diffeomorphisms.
There is a natural sheaf of algebras on this site given by smooth functions $C^\infty : M \mapsto C^\infty(M)$. 
%By definition, a $C^\infty$-module is a module sheaf over $C^\infty$ on ${\rm Mfld}$. 

For any $p$ the assignment $\Omega^p : M \mapsto \Omega^p (M)$ defines a $C^\infty$-module.
Similarly, if $F$ is any $C^\infty$-module we have the $C^\infty$-module of $p$-forms with values on $F$ defined by the assignment 
\ben
\Omega^1(F) : M \in {\rm Mfld} \mapsto \Omega^p(M, F) = \Omega^p(M) \tensor_{C^\infty(M)} F(M) .
\een

\begin{dfn}
A {\em differentiable vector space} is a $C^\infty$-module equipped with a map of sheaves on ${\rm Mfld}$
\ben
\nabla : F \to \Omega^1(F) 
\een 
such that for each $M$, $\nabla(M)$ defines a flat connection on the $C^\infty(M)$-module $F(M)$. 
A map of differentiable vector spaces is one of $C^\infty$-modules that intertwines the flat connections. 
This defines a category that we denote ${\rm DVS}$.
\end{dfn}

Our favorite example of differentiable vector spaces are imported directly from geometry.

\begin{eg}
Suppose $E$ is a vector bundle on a manifold $X$. 
Let $\sE(X)$ denote the space of smooth global sections.
Let $C^\infty(M, \sE(X))$ be the space of sections of the bundle $\pi_X^*E$ on $M \times X$ where $\pi_X : M \times X \to X$ is projection. 
The assignment $M \mapsto C^\infty(M , \sE(X))$ is a $C^\infty$-module with flat connection, so defines a differentiable vector space.
Similarly, the space of compactly supported sections $\sE_c(X)$ is a DVS. 
\end{eg}

Many familiar categories of topological vector spaces embed inside the category of differentiable vector spaces. 
Consider the category of locally convex topological vector spaces ${\rm LCTVS}$.
If $V$ is such a vector space, there is a notion of a smooth map $f : U \subset \RR^n \to V$.
One can show, Proposition B.3.0.6 of \cite{CG1}, that this defines a functor ${\rm dif}_t : {\rm LCTVS} \to {\rm DVS}$ sending $V$ to the $C^\infty$-module $M \mapsto C^\infty(M, V)$.
If ${\rm BVS} \subset {\rm LCTVS}$ is the subcategory with the same objects but whose morphisms are bounded linear maps, this functor restricts to embed ${\rm BVS}$ as a full subcategory ${\rm BVS} \subset {\rm DVS}$. 

There is a notion of completeness that is useful when discussing tensor products. 
A topological vector space $V \in {\rm BVS}$ is {\em complete} if every smooth map $c : \RR \to V$ has an anti-derivative \cite{KM97}.
There is a full subcategory ${\rm CVS} \subset {\rm BVS}$ of complete topological vector spaces.
The most familiar example of a complete topological vector space will be the smooth sections $\sE(X)$ of a vector bundle $E \to X$.

We let ${\rm Ch}({\rm DVS})$ denote the category of cochain complexes in differentiable vector spaces (we will refer to objects as differentiable vector spaces). 
It is enriched over the category of differential graded vector spaces in the usual way.
We say that a map of differentiable cochain complexes $f : V \to W$ is a quasi-isomorphism if and only if for each $M$ the map $f : C^\infty(M, V) \to C^\infty(M,W)$ is a quasi-isomorphism.

\begin{thm}[Appendix B \cite{CG1}]
The full subcategory ${\rm dif}_c : {\rm CVS} \subset {\rm DVS}$ is closed under limits, countable coproducts, and sequential colimits of closed embeddings. 
Furthermore, {\rm CVS} has the structure of a symmetric monoidal category with respect to the completed tensor product $\Hat{\tensor}_{\beta}$. 
\end{thm}

We will not define the tensor product $\Hat{\tensor}_\beta$here, but refer the reader the cited reference for a complete exposition.
We will recall its key properties below.
Often times we will write $\tensor$ for $\Hat{\tensor}_\beta$ where there is no potential conflict of notation. 
The fundamental property of the tensor product that we use is the following.
Suppose that $E,F$ are vector bundles on manifolds $X,Y$ respectively.
Then, $\sE(X), \sF(Y)$ lie in ${\rm CVS}$, so it makes sense to take their tensor product using $\Hat{\tensor}_\beta$. 
There is an isomorphism
\be\label{tensor1}
\sE(X) \Hat{\tensor}_\beta \sF(Y) \cong \Gamma(X \times Y, E \boxtimes F)
\ee
where $E \boxtimes F$ denotes the external product of bundles, and $\Gamma$ is smooth sections. 

If $E$ is a vector bundle on a manifold $X$, then the spaces $\sE(X), \sE_c(X)$ both lie in the subcategory ${\rm CVS} \subset {\rm DVS}$. 
The differentiable structure arises from the natural topologies on the spaces of sections. 

We will denote by $\Bar{\sE}(X)$ ($\Bar{\sE}_c(X)$) the space of (compactly supported) distributional sections.
It is useful to bear in mind the following inclusions
\ben
\begin{tikzcd}      
& \Bar{\sE}_c(X) \arrow[hook]{dr} & \\
\sE_c(X)  \arrow[hook]{ur} \arrow[hook]{rd} & & \Bar{\sE}(X) \\
    & \sE(X) \arrow[hook]{ru} & .
\end{tikzcd}
\een
When $X$ is compact the bottom left and top right arrows are equalities.

Denote by $E^\vee$ the dual vector bundle whose fiber over $x \in X$ is the linear dual of $E_x$. 
Let $E^!$ denote the vector bundle $E^\vee \tensor {\rm Dens}_X$, where ${\rm Dens}_X$ is the bundle of densities. 
In the case $X$ is oriented, ${\rm Dens}_X$ is isomorphic to the top wedge power of $T^*X$. 
Let $\sE^!(X)$ denote the space of sections of $E^!$. 
The natural pairing 
\ben
\sE_c(X) \tensor \sE^!(X) \to \CC
\een
that pairs sections of $E$ with the evaluation pairing and integrates the resulting compactly supported top form exhibits $\Bar{\sE}_c(X)$ as the continuous dual to $\sE^!(X)$. 
Likewise, $\sE_c(X)$ is the continuous dual to $\Bar{\sE}^!(X)$. 
In this way, the topological vector spaces $\Bar{\sE}(X)$ and $\Bar{\sE}_c(X)$ obtain a natural differentiable structure.

If $V$ is any differentiable vector space then we define the space of linear functionals on $V$ to be the space of maps $V^* = {\rm Hom}_{\rm DVS}(V, \RR)$. 
Since ${\rm DVS}$ is enriched over itself this is again a differentiable vector space. 
Similarly, we can define the polynomial functions of homogeneous degree $n$ to be the space
\ben
\Sym^n(V^*) = {\rm Hom}^{multi}_{\rm DVS}(V \times \cdots \times V, \RR)_{S_n}
\een
where the hom-space denotes multi-linear maps, and we have taken $S_n$-coinvariants on the right-hand side.
The algebra of functions on $V$ is defined by
\ben
\sO(V) = \prod_{n} \Sym^n(V^*) .
\een

As an application of Equation (\ref{tensor1}) we have the following identification.

\begin{lem}\label{lem: fnls}
Let $E$ be a vector bundle on $X$. 
Then, there is an isomorphism
\ben
\sO(\sE(X)) \cong \prod_{n} \sD_c(X^n , (E^!)^{\boxtimes n})_{S_n}
\een
where $\sD_c(X^n , (E^!)^{\boxtimes n})$ is the space of compactly supported distributional sections of the vector bundle $(E^!)^{\boxtimes n}$.
Again, we take $S_n$-coinvariants on the right hand side.
\end{lem}

\bibliographystyle{alpha}
\bibliography{holrenorm}

\end{document}